\documentclass[11pt, a4paper]{article}

\usepackage[english]{babel}
\usepackage{enumitem}
\usepackage{amsmath,mathtools}
\usepackage{amsthm}
\usepackage{amssymb}
\usepackage{xcolor}
\usepackage{algorithm}
\usepackage{algpseudocode}
\usepackage{tikz}
\usepackage{cite}
\usepackage[blocks]{authblk}
\usepackage{subcaption}
\usepackage{dsfont}
\usepackage{microtype}
\usepackage{todonotes}
\usepackage{thmtools}
\usepackage{thm-restate}
\usepackage{array}
\usepackage{multirow}
\usepackage{booktabs}
\usepackage[margin=1in]{geometry}
\newcommand\numberthis{\addtocounter{equation}{1}\tag{\theequation}}
\definecolor{myred}{RGB}{227, 38, 54}
\definecolor{myblue}{RGB}{75, 141, 200} 
\definecolor{mygreen}{RGB}{87, 166, 74} 
\definecolor{greenish}{RGB}{27,158,119}
\definecolor{MyOrange}{RGB}{217,95,2}
\definecolor{MyPurple}{RGB}{117,112,179}
\usepackage{hyperref}
\hypersetup{
	colorlinks,
	linkcolor={red!60!black},
	citecolor={blue!60!black},
	urlcolor={blue!80!black}
}
\usepackage[nameinlink]{cleveref}

\declaretheorem[name=Theorem, numberwithin=section]{thm}
\declaretheorem[name=Lemma, numberlike=thm, numberwithin=section]{lemma}

\declaretheorem[name=Corollary, numberlike=thm, numberwithin=section]{corollary}

\declaretheorem[name=Example, style=definition, numberwithin=section]{example}

\crefname{subsection}{subsection}{subsections}
\crefname{example}{example}{examples}

\title{Combinatorial Perpetual Scheduling: \\ Existence and Computation of Low-Height Schedules}
\author{
    Mirabel Mendoza-Cadena\thanks{Centro de Modelamiento Matemático (CNRS IRL2807), Universidad de Chile, Santiago, Chile. \texttt{lmmendoza@cmm.uchile.cl}.}\qquad  
    Arturo Merino\thanks{Department of Computer Science, University of Chile, Chile. \texttt{amerino@dcc.uchile.cl}.}\qquad\\
    Mads Anker Nielsen\thanks{Department of Mathematics and Computer Science, University of Cologne, Germany. \texttt{m.nielsen@uni-koeln.de}.}\qquad
    Kevin Schewior\thanks{Department of Mathematics and Computer Science, University of Cologne, Germany, and Department of Mathematics and Computer Science, University of Southern Denmark, Denmark. \texttt{schewior@cs.uni-koeln.de}.}
}
\date{\today}

\newcommand{\R}{\mathbb{R}}
\newcommand{\N}{\mathbb{N}}
\newcommand{\PP}{\mathbb{P}}
\newcommand{\QQ}{\mathbb{Q}}

\newcommand{\EE}{\mathbb{E}}
\newcommand{\NN}{\mathbb{N}}
\newcommand{\I}{\mathcal{I}}
\newcommand{\F}{\mathcal{F}}
\newcommand{\M}{\mathcal{M}}
\newcommand{\cI}{\mathcal{I}}

\newcommand{\cH}{\mathcal{H}}

\DeclareMathOperator{\conv}{conv}

\newcommand{\hatcirc}{\mathbin{\hat\circ}}

\newcommand{\fl}[1]{\lfloor #1 \rfloor}
\newcommand{\ce}[1]{\lceil #1 \rceil}

\def\final{0}  
\def\iflong{\iffalse}
\ifnum\final=0  
\newcommand{\mnote}[1]{{\color{violet!70}[{\tiny \textbf{Mirabel:} \bf #1}]\marginpar{\color{violet!70}*}}}

\else 
\newcommand{\mnote}[1]{}
\fi

\newcommand{\ind}{\mathds{1}}
\begin{document}
\maketitle

\begin{abstract}
This paper considers a framework for combinatorial variants of perpetual-scheduling problems. 
Given an independence system $(E,\mathcal{I})$, a schedule consists of an independent set $I_t \in \mathcal{I}$ for every time step $t \in \mathbb{N}$, with the objective of fulfilling frequency requirements on the occurrence of elements in $E$. 
We focus specifically on \emph{combinatorial bamboo garden trimming}, where elements accumulate height at growth rates $g(e)$ for $e \in E$
and are reset to zero when scheduled, with the goal of minimizing the maximum height attained by any element. We assume that $g$ is normalized so that it is a convex combination of the incidence vectors of $\mathcal{I}$.

Using the integrality of the matroid-intersection polytope, we prove that, when $(E,\mathcal{I})$ is a matroid, it is possible to guarantee a maximum height of at most 2, which is optimal. 
We complement this existential result with efficient algorithms for specific matroid classes, achieving a maximum height of 2 for uniform and partition matroids, and 4 for graphic and laminar matroids. 
In contrast, we show that for general independence systems, the optimal guaranteed height is $\Theta(\log |E|)$ and can be achieved by an efficient algorithm. 
For \emph{combinatorial pinwheel scheduling}, where each element $e\in E$ needs to occur in the schedule at least every $a_e \in \mathbb{N}$ time steps, our results imply bounds on the density sufficient for schedulability.
\end{abstract}

\section{Introduction}

Certain tasks such as maintenance, monitoring, or cleaning do not merely require to be accomplished once. Instead, they often require to be executed again and again under a certain frequency requirement, virtually forever.
Such tasks are captured by perpetual-scheduling problems with frequency requirements, a class of problems which has recently seen a surge of interest. These problems, often appealingly clean, include the pinwheel scheduling (PS) problem (and variants)~\cite{Kawamura24,KanellopoulosKMPP26,Mishra26} and the bamboo-garden trimming (BGT) problem (and variants)~\cite{Kuszmaul22,HohneS23,BiktairovGJNSW25}. 
Kawamura's survey~\cite{Kawamura2025} gives a broad overview of the state of the art in this area.

To a large extent, the literature has only considered settings in which, at any time, only a single task can be scheduled.
In many applications, however, it may be possible to schedule multiple tasks at the same time, subject to a combinatorial constraint.
For instance, there may be multiple workers, each of which has the expertise to do a certain subset of the tasks. In this case, a set of tasks can be executed simultaneously if and only if it is an independent set of a transversal matroid.
In this paper, we study a general framework for combinatorial variants of perpetual scheduling, specifically, PS and BGT. This framework essentially (cf.\ discussion later) unifies the few combinatorial models considered in the literature~\cite{Bar-NoyL03,ImMZ21,GasieniecSW24}. The framework was introduced in \cite{cicerone2019FairHitting} as the \emph{fair hitting sequence problem}. 
However, their work studies the problem from an approximation-algorithms perspective, which is not the main focus of the present paper.

The majority of this paper is written in terms of what we call \emph{combinatorial BGT} (CBGT).
Here, there is a ground set $E$ of $n$ \emph{bamboos} and a fixed independence system $\mathcal{I}\subseteq 2^E$. Initially, each bamboo $e\in E$ has height $0$ and a fixed \emph{growth rate} $g(e)\in\mathbb{Q}_{\geq 0}$. 
The following two steps alternate ad infinitum: First, all bamboos grow according to their growth rates. Then, an independent set $I\in \mathcal{I}$ can be selected to be cut, i.e., all heights of bamboos in $I$ are reduced to $0$. This process leads to an infinite valid \emph{schedule} $\pi=I_1,I_2,\ldots$, with $I_i\in\mathcal{I}$. The goal is to minimize the maximum height $h(\pi)$ that \emph{ever} occurs among any of the bamboos. When $\mathcal{I}$ is the $1$-uniform matroid, this is precisely vanilla BGT.

Note that scaling $g$ by some factor $c$ also just scales the height $h(\pi)$ of any schedule $\pi$ by $c$. For this reason, it is standard for BGT to assume that $g$ is \emph{normalized} so that $\sum_{e\in E}g(e)=1$. For CBGT, we generalize this by assuming that $g$ is a convex combination of the incidence vectors of $\I$. That is, there exists coefficients $\lambda(I) \geq 0$ for each $I \in \I$ such that $\sum_{I \in \I} \lambda(I) = 1$ and $g = \sum_{I \in \I} \lambda(I)\ind_I$, where $\ind_I \in \{0,1\}^E$ is the vector such that $\ind_I(e) = 1$ iff $e \in I$ for all $e \in E$. This is essentially (again cf.\ discussion later) the type of normalization used in the combinatorial models in the literature~\cite{ImMZ21,GasieniecSW24}.

The main result of this paper is that, when $(E,\mathcal{I})$ is a matroid, a maximum height of at most $2$ can be guaranteed. The bound is tight (as can be easily seen from an instance with $n=2$ and growth rates $1-\varepsilon,\varepsilon$ for vanilla BGT~\cite{GasieniecJKLLMR24}) and generalizes the same bound for vanilla BGT~\cite{GasieniecJKLLMR24}. While our proof only yields a pseudo-polynomial-time algorithm that produces any next cut, we give polynomial-time algorithms for uniform matroids, graphic matroids, and laminar matroids, in the latter two cases only guaranteeing a maximum height of at most $4$. When $k=1$, our algorithm for $k$-uniform matroids implies a simplification of the state of the art (that has appeared before in a different context~\cite
{bar-noyEfficientPeriodicScheduling2002}). We also show that, for general set systems $\mathcal{I}$, the smallest height that can be guaranteed is $\Theta(\log |E|)$. We also give a deterministic polynomial-time algorithm implementing a height-$O(\log |E|)$ schedule, which implies an $O(\log |E|)$-approximation algorithm, solving an open problem from~\cite{cicerone2019FairHitting}.

In combinatorial PS (CPS), again an independence system $(E,\mathcal{I})$ is given, and each element $e$ of the ground set $E$ has a \emph{period} $a(e)\in\mathbb{N}$. An instance is called \emph{schedulable} if there exists an infinite valid schedule $I_1,I_2,\ldots$ with $I_i\in\mathcal{I}$ such that all periods are satisfied, i.e., for each $e\in E$ and $t\in\mathbb{N}$, it holds that $e\in \bigcup_{i=t}^{i+a(e)-1}I_i$. Analogously to the above normalization, we define the \emph{density} of an instance as the minimum $\rho$ such that there exist $\lambda(I)\geq 0$ for all $I\in\mathcal{I}$ with $\sum_{i\in\mathcal{I}}\lambda(I)=\rho$ and $\sum_{I\in\mathcal{I}:e\in I} \lambda(I) = 1/a(e)$ for all $e\in E$. Our results for CBGT imply \emph{density bounds}, i.e., bounds on the density that guarantees schedulability, the quantity that has been the main object of study for vanilla PS~\cite{Kawamura2025}.

\subsection{Our Contribution}

In this exposition, we focus on CBGT. Our main result is the following. 

\begin{restatable*}{thm}{matroidtwo}
    \label{thm:matroidtwo}    
    For any CBGT instance $(E,\I,g)$ where $(E,\I)$ is a matroid, there exists an infinite valid schedule $\pi$ for $(E,\I,g)$ with
    $h(\pi) < 2$.
\end{restatable*}

A first approach could be the following: Reduce the problem to vanilla BGT by identifying the independent sets $I\in\I$ with bamboos in the new instance and interpreting the coefficients $\lambda(I)$ for all $I\in\I$ as growth rates. Then an optimal valid schedule $\pi'$ for the vanilla-BGT instance can be interpreted as a valid schedule $\pi$ for the CBGT instance. Unfortunately, this leads to an arbitrarily large height, even when $(E,\I)$ is a uniform matroid: Let $n=|E|$ be even, let $\I=\{I\subseteq E\mid |I|\leq n/2\}$, and set $\lambda(I)=1/\binom{n}{n/2}$ for all $I\in\I$ with $|I|=n/2$ (and $\lambda(I)=0$ for all other $I$), which implies $g(e)=1/2$ for all $e\in E$. Note that, on this instance, $\pi$ starts with an arbitrary permutation of $\mathcal{I}$. In particular, for some fixed $e^\star$, $\pi$ may first schedule all $\binom{n}{n/2}/2$ sets that do not contain $e^\star$, leading to a height of $\binom{n}{n/2}/4$. Still, even a height-1 schedule for this instance exists, as achieved by just alternating between any set and its complement. Thus, it is inevitable to take into account the structure of the set system rather than just the values of~$\lambda$, which seems to be challenging in general.

To establish our result, we consider a slightly different setting than CBGT, which is also of independent interest: In BGT terms, upon cutting, the height of a bamboo is not reduced \emph{to $0$} but \emph{by $1$}, possibly below $0$. We call the supremum of all \emph{absolute} heights (positive or negative) any bamboo ever reaches the \emph{discrepancy} $d(\pi)$ of the corresponding schedule $\pi$. Later, we give a formal definition, and show that $d(\pi)<1$ implies $h(\pi)<2$ for all schedules $\pi$.
(It is also not difficult to see that our result implies a height bound of $2$ in the setting where one may cut by at most $1$ and never below $0$; see, e.g.,~\cite[Section 4]{Kuszmaul22} for this setting.)

For this variant, our normalization (i.e., that $g$ lies in the matroid polytope) is necessary to obtain a schedule of finite discrepancy. Indeed, suppose that an infinite schedule $\pi$ has finite $d(\pi)$. As $g$ is rational, each $e \in E$ can only take finitely many distinct heights without its absolute height exceeding $d(\pi)$, and thus we may assume that $\pi$ eventually cyclically repeats some finite sequence $\pi'$. Each $e \in E$ must be cut by exactly as much as it grows in $\pi'$, as otherwise $d(\pi)$ would diverge to infinity. Thus, taking $\lambda(I)$ to be the fraction of time steps $t \in [|\pi'|]$ such that $\pi'(t) = I$ for all $I \in \I$, we have $\sum_{I \in \I} \lambda(I) = 1$ and $\sum_{I \in \I} \lambda(I)\ind_I = g$, which by definition implies that $g$ lies in the matroid polytope of $(E,\I)$.

Note that, if one could cut according to some point in the matroid polytope, one would simply be able to remove precisely the just-grown height. Minimizing discrepancy can thus also be understood as the problem of rounding the fractional growth process to an integral sequence of cuts. 
In other words, the question is: \emph{How well can a continuous combinatorial process be approximated through a discrete one?} (See~\cite{ImMZ21} for stating a similar question in the context of the matroidal adversarial version.)

The crucial ingredient is the following. Fix some $e\in E$ and $T\in\N$ such that $Tg(e)\in\N$. Then the sets $B\subseteq\{1,\dots,T\}$ of time points at which $e$ can be cut so as to fulfill the discrepancy requirement are precisely the bases of a matroid.
Not only that, the $T$-dimensional vector whose entries are identical $g(e)$ turns out to be contained in the corresponding matroid polytope.

This allows us to present the first formulation of BGT-like problems in terms of matroid intersection, which may be of independent interest. To this end, we define
two matroids on the ground set $E\times[T]$ (for some $T$ such that $Tg(e)\in\N$ for all $e\in E$): The direct sum $M_T$ of $T$ copies of the input matroid, and the direct sum $M_E$ of the $|E|$ matroids that we just described. By our normalization and the above ingredient, we get that the vector $y^\star\in\R^{E\times[T]}$ with $y_{(e,t)}=g(e)$ for all $(e,t)\in E\times[T]$ is contained in the \emph{matroid-intersection} polytope of $M_T$ and $M_E$. This allows us to use the well-known integrality of the matroid-intersection polytope~\cite{Edmonds1970} to show that $M_T$ and $M_E$ have a common basis, an infinite repetition of which corresponds to our desired infinite schedule with discrepancy less than $1$.

Ideally, we would like to have a polynomial-time algorithm for computing any next cut of such schedule. While mathematically clean, our above proof does not yield that; indeed, it seems inherently pseudo-polynomial. To address this drawback, we give polynomial-time algorithms for certain classes of matroids. We show how \emph{tree schedules}~\cite{bar-noyEfficientPeriodicScheduling2002} can be used to get height-2 schedules on uniform (and thus partition) matroids in polynomial time. Then, using the concept of \emph{rainbow-circuit-free colorings}~\cite{fujita2010rainbow,Gouge10,hoffman2019rainbow,jendrol200rainbow}, we give polynomial-time algorithms for graphic and laminar matroids achieving height 4.

One may wonder whether \Cref{thm:matroidtwo} extends beyond matroids. Certainly, there are non-matroid independence systems such as $(\{a,b,c\},\{\emptyset,\{a\},\{b\},\{c\},\{b,c\}\})$, for which a height of $2$ is also achievable. Indeed, here $b$ and $c$ can simply be treated as a single element from the ground set. We, however, establish that there are set systems for which not even a constant height is achievable.

\begin{restatable*}[$\star$]{thm}{jamisonlb}
    \label{thm:jamison-lower-bound}
    For every $k \in \NN$ there is a CBGT instance with $n=2^k-1$ elements where every infinite valid schedule has height at least $\frac{1}{2}\log_2 (n+1) - \frac{1}{2}$. 
\end{restatable*}

While the proof of this theorem requires algebraic methods, we also give a simpler construction (in the proof of \Cref{thm:lower}) that uses first principles but loses constant factors. 
We show that there exists an instance $(E,\mathcal{I},g)$ in which $g(e)=1/2$ for all $e\in E$ such that, for every subfamily $\mathcal{I}'\subseteq \mathcal{I}$ of sufficiently small size (but even of size $\Omega(\log |E|)$), there exists an element $e\in E$ with $e\notin \bigcup_{I\in\mathcal{I}'} I$.
Consequently, if $\mathcal{I}'$ constitutes the first cuts in a schedule, some bamboo is never cut and therefore reaches height $\Omega(\log |E|)$.

Finally, we show that this upper bound can be matched asymptotically, even with a polynomial-time algorithm (assuming access to a maximum-weight independent-set oracle over the set system).
Assuming that $\I$ is given explicitly, a polynomial-time algorithm computing an $O(\log |E|)$-height schedule with high probability was given in \cite{cicerone2019FairHitting}. Our algorithm answers the open question of whether there exists a deterministic polynomial-time algorithm.

\begin{restatable*}{thm}{compub}
    \label{thm:compub}
    For any CBGT instance $(E,\cI,g)$ with $n$ elements there is an infinite valid schedule $\pi$ of period $2n$ such that $h(\pi)\leq O(\log n)$.
    Furthermore, this schedule can be computed by making at most $2n$ queries to a maximum-weight independent-set oracle for $\I$.
\end{restatable*}

In fact, our algorithm implements a type of schedule that we can also obtain through a probabilistic argument combined with a straightforward case distinction into fast and slow bamboos. This technique was already used in~\cite{cicerone2019FairHitting} for the aforementioned result. We obtain the polynomial-time algorithm by choosing weights for the maximum-weight independent-set oracle such that we greedily minimize a carefully chosen potential function at each step.

\begin{table}[t]
    \centering
    \begin{tabular}{r>{\centering}m{3cm}>{\centering}m{3cm}>{\centering\arraybackslash}m{3cm}}
    \toprule
\multirow{2}{*}{}&\multicolumn{2}{c}{upper bound}&\multirow{2}{*}{lower bound}\\\cmidrule{2-3}
&polynomial-time&existence&\\ \midrule
        general set systems & $O(\log |E|)$\hspace{3cm}{\scriptsize (\Cref{thm:compub})} & $O(\log |E|)$\hspace{3cm}{\scriptsize (\hspace{1sp}\cite{cicerone2019FairHitting}, \Cref{thm:exub})} & $\Omega(\log |E|)$\hspace{3cm}{\scriptsize (\Cref{thm:jamison-lower-bound})} \\[.5cm]

        general matroids & -- & $2$\hspace{3cm}{\scriptsize (\Cref{thm:matroidtwo})} & $2$\hspace{3cm}{\scriptsize (1-uniform example)} \\[.5cm]
        
        graphic matroids & $4$\hspace{3cm}{\scriptsize (\Cref{thm:graphic})} & $2$ \hspace{3cm}{\scriptsize ($\uparrow$)} & $2$\hspace{3cm}{\scriptsize (1-uniform example)}\\[.5cm]

        laminar matroids & $4$\hspace{3cm}{\scriptsize (\Cref{thm:laminar-FUN})} & $2$ \hspace{3cm}{\scriptsize ($\uparrow$)} & $2$\hspace{3cm}{\scriptsize (1-uniform example)}\\[.5cm]

        uniform matroids & $2$\hspace{3cm}{\scriptsize (\Cref{lem:kuniform})} & $2$ \hspace{3cm}{\scriptsize ($\leftarrow$, $\uparrow$)} & $2$\hspace{3cm}{\scriptsize (1-uniform example)}\\[.5cm]
        
        partition matroids & $2$\hspace{3cm}{\scriptsize (\Cref{cor:partition})} & $2$ \hspace{3cm}{\scriptsize ($\leftarrow$, $\uparrow$)} & $2$\hspace{3cm}{\scriptsize (1-uniform example)}\\
        \bottomrule
    \end{tabular}
    \caption{Overview of our results. When result follow from other, stronger results, this is marked with an arrow in the direction of the respective stronger result.}
    \label{tab:results}
\end{table}

We summarize our results in \Cref{tab:results}.

\subsection{Further Related Work}

We start by discussing vanilla (i.e., 1-uniform) perpetual-scheduling problems. In the paper that introduced BGT~\cite{GasieniecJKLLMR24}, it was shown by using a density bound (even a suboptimal one~\cite{FL02}) for vanilla PS that an infinite schedule achieving height at most $2$ always exists.
The same bound can also be achieved by a simpler deadline-driven strategy~\cite{Kuszmaul22}. 
Only very recently, it was disproved that the Greedy algorithm, which at any time cuts the tallest bamboo, achieves a height of at most $2$, by showing a lower bound of approximately $2.07$~\cite{GreedyLB25}; an upper bound of $4$ is known~\cite{Kuszmaul22}.

Another line of work~\cite{Croce2020,Ee21,HohneS23,Kawamura24,Mishra26} is concerned not with bounding the achievable height across all instances but with, given an instance, efficiently computing a schedule that approximates the lowest height across all schedules for that instance. 
The current state of the art in this setting is a $9/7$-approximation~\cite{Mishra26}.
For CBGT, note that, when the minimum sum of coefficients to generate $g$ is precisely $1$, the height of any schedule is at least $1$. Therefore, our bounds also imply results of the aforementioned flavor for our combinatorial variant.

Variants in which the growth rates may vary over time have also been considered. In cup games~\cite{AdlerBFGGP03,BenderFK19,BenderK21}, the growth rate is chosen online by an adversary, and it is known that Greedy achieves a discrepancy of at most $O(\log n)$ for a ground set of size $n$. Online proportional apportionment~\cite{CembranoCGV26} can be interpreted as the variant of this in which one may only cut bamboos that have just grown. In chairperson assignment~\cite{Tijdeman1980}, the varying growth rates are known ahead of time, and it is known~\cite{Tijdeman1980} that a discrepancy of less than $1$ can be achieved; a weighted version has been considered recently~\cite{LiuR26}. We note that the result for the unweighted problem together with our observation on the relationship between discrepancy and height (\Cref{lem:disc1height}) also implies that a height-$2$ schedule for vanilla BGT exists, a fact that seems to have been missed by the literature.

For vanilla PS, as introduced in the 1990s~\cite{HRTV92}, the density bound of $5/6$ was recently confirmed in a  breakthrough paper by Kawamura~\cite{Kawamura24}. A tight example (periods $2$, $3$, and $x\in\mathbb{N}$) and several weaker bounds~\cite{CC93, LL97, FL02} had previously been known. Whether deciding feasibility is in $\mathsf{P}$ remains, however, open~\cite{Kawamura2025}. Recent works consider a finite version of vanilla PS~\cite{KanellopoulosKMPP26} as well as a covering version~\cite{Mishra26}.

Combinatorial variants where the independence system is a $k$-uniform matroid have been considered for pinhweel scheduling in form of a problem called \emph{windows scheduling}~\cite{Bar-NoyL03}, where it has been shown (implicitly) that the density threshold approaches $1$ as $k\to\infty$. It has also been considered for cup games, where the achievable backlog is known to be in $O(\log n)\cap\Omega(\log (n-k))$~\cite{Kuszmaul20}. General matroids, with the same type of normalization, have also been considered for cup games~\cite{ImMZ21}, where the upper bound of $O(\log n)$ can be recovered.

In so-called \emph{polyamorous scheduling}~\cite{GasieniecSW24,BiktairovGJNSW25}, (essentially) a special case of CBGT, the ground set is the set of edges of a graph, and the set of independent sets is the set of matchings. For normalization, the authors assume that the vector of growth rates is contained in the \emph{fractional} matching polytope. This type of normalization is completely identical to our normalization on bipartite graphs (due to the integrality of said polytope) and identical up to a factor of $3/2$ on general graphs. It is known that a height of $4$ is achievable, and a height of $5.24$ in polynomial time~\cite{BiktairovGJNSW25}. The authors also show approximation hardness and consider a pinwheel version of that problem.

A generalization of the perpetual scheduling problem to set systems was previously proposed in~\cite{cicerone2019FairHitting} under the name \textit{Fair Hitting Sequence Problem}. In this problem, the input is a triple $(A,\mathcal{F},g')$ where $(A,\mathcal{F})$ is a set system as for CBGT. However, the set system is not necessarily downward closed, and the growth rates are on the sets of $\mathcal{F}$. A schedule is an infinite sequence of elements $a \in A$, and all sets $F \in \mathcal{F}$ containing $a$ have their height reset to $0$ upon cutting $a$. If we set $E = \mathcal{F}$, $\I = \{\{F \in \mathcal{F} \mid a \in F\} \mid a \in A\}$, and $g = g'$, then any schedule for $(A,\mathcal{F},g')$ corresponds to a schedule for $(E,\I,g)$ with the same height, and thus the models are equivalent (up to downward closedness and normalization). While we focus on the existence and computation of low-height schedules, the authors of \cite{cicerone2019FairHitting} focus on computing an approximation of the optimal schedule. They show that the exact computation of an optimal schedule is NP-hard, and also obtain a randomized polynomial-time $O(\log |E|)$-approximation algorithm via the rounding of an LP. It is not too difficult to see that their result actually implies the existence of an $O(\log |E|)$-height schedule with our normalization assumption.
In \Cref{ss:log-lower-bound}, we give an alternative proof of the existence of a $O(\log |E|)$-height schedule and in \Cref{ss:logboundimplement}, we give a deterministic polynomial-time algorithm which computes such a schedule. This answers an open question from \cite{cicerone2019FairHitting}.

\subsection{Overview}

In \Cref{sec:prelim}, we give definition and notation, and we also establish the relationship between CBGT and CPS. In \Cref{sec:matroid-height2}, we prove the main result. In \Cref{sec:FUN-algorithms}, we give the polynomial-time algorithms for special matroids. In \Cref{sec:general}, we consider general independence systems. We state open problems in \Cref{sec:open}. Proofs of results marked with $\star$ can be found in the appendix.

\section{Preliminaries}
\label{sec:prelim}

We give the definitions used throughout the paper, repeating some of the definitions from the introduction. We first define $\N = \{1,2,\dots\}$. We denote by $[n]$ the set
$\{1,2,\dots,n\}$ for any $n \in \N$.
We denote by $\QQ_{\geq 0}$ the set of nonnegative rational numbers.

\paragraph{Combinatorial Bamboo Garden Trimming.} 
An instance of \emph{Combinatorial Bamboo Garden Trimming} (CBGT) is a triple $(E,\I,g)$ where $(E,\I)$ is an arbitrary set system and $g \in \QQ_{\geq 0}^E$ is a growth rate vector. 
We require that $g$ lies in the convex hull of the characteristic vectors of the independent sets in $\I$. 
That is, there exists a set of non-negative coefficients $\{\lambda(I) \mid I \in \mathcal{I}\}$ such that $\sum_{I \in \I} \lambda(I) = 1$ and $g(e) = \sum_{I \in \mathcal{I} : e \in I} \lambda(I)$ for every $e \in E$.

A \emph{schedule} $\pi$ for a CBGT instance $(E,\I,g)$ is a finite or infinite sequence $\pi(1),\pi(2),\dots$ where $\pi(t) \subseteq E$ is the set of bamboos to cut on day $t$ for every $t \in [|\pi|]$. 
If $\pi$ is infinite, then we define $[|\pi|] = \N$.
A schedule is \emph{valid} if $\pi(t) \in \I$ for every $t \in [|\pi|]$ and otherwise it is \emph{invalid}. 
The following definitions apply to \emph{all} (finite/infinite valid/invalid) schedules.

For any $t \in [|\pi|]$ and $e \in E$, we denote by $h^\pi_t(e)$ the height of bamboo $e$ on day $t$ after $e$ has grown but before $\pi(t)$ is cut.
The \emph{height} $h(\pi)$ of a schedule $\pi$ is then 
\[
    h(\pi) = \max_{e \in E}\sup_{t \in [|\pi|]} h^\pi_t(e)
\] 
or, in words, the maximum height reached by any bamboo during the execution of $\pi$. 
We can, equivalently, define the height of a schedule $\pi$ as follows. 
The \emph{recurrence time} $\gamma^\pi(e)$ of a bamboo $e \in E$ under a schedule $\pi$ is the maximum number of days between consecutive schedulings of $e$ according to $\pi$. 
That is, $\gamma^\pi(e) = \sup_{i,i' \in [|\pi|]}\{i'-i+1 \mid e \notin \bigcup_{j=i}^{i'} \pi(j) \}$.
The maximum height $h^\pi(e)$ reached by a bamboo $e \in E$ under $\pi$ is then $h^\pi(e) = \gamma^\pi(e) \cdot g(e)$ and $h(\pi) = \max_{e \in E} h^\pi(e)$.

We also define the \emph{discrepancy} $d(\pi)$ of a schedule $\pi$. 
Intuitively, the discrepancy is a measure of how closely the fraction of the time that $e$ has been scheduled until day $t$ matches the ideal fraction $t\cdot g(e)$. 
Define $A(\pi,e,t) = |\{1 \leq i \leq t \mid e \in \pi(i)\}|$ for all $t \in [|\pi|]$ as the number of occurrences of $e$ among the first $t$ sets of $\pi$. 
The discrepancy $d^\pi(e)$ of a bamboo $e \in E$ is then
\[
    d^\pi(e)  = \sup_{t \in [|\pi|]} |t \cdot g(e) - A(\pi,e,t)|,
\]
and we define the discrepancy $d(\pi)$ of $\pi$ as $d(\pi) = \max_{e \in E}d^\pi(e)$.
Note that $d^\pi(e) < 1$ if and only if $\fl{tg(e)} \leq A(\pi,e,t) \leq \ce{tg(e)}$ for $t \in [|\pi|]$.

\paragraph{Graph theory.} 
Let $G = (V(G),E(G))$ be a graph. 
Given a subset $X \subseteq V(G)$, the \emph{subgraph of $G$ induced by $X$} is denoted by $G[X]$ and is obtained from $G$ by removing all vertices in $V(G) \setminus X$ along with any edges incident to them. 
The \emph{neighborhood} $N_G(v)$ of a vertex $v \in V(G)$ is the set of vertices adjacent to $v$ in $G$.

\paragraph{Independence Systems and Matroids.}
We recall some basic definitions from matroid theory and refer the reader to~\cite{oxley2011matroid} for further details.

We denote by $\R^E$ the set of $|E|$-dimensional real vectors indexed by $e \in E$. 
We denote by $\mathbf{1}^n$ the $n$-dimensional vector whose entries are all $1$. 
We omit the superscript $n$ when it is clear from context. 
The \emph{convex hull} of a set $X \subseteq \R^n$ is denoted $\conv(X)$. 
The \emph{incidence vector} of a subset $X \subseteq E$ is denoted by $\ind_X$ and is defined as the $|E|$-dimensional vector such that $\ind_X(e) = 1$ if $e \in X$ and otherwise $\ind_X(e) = 0$ for all $e \in E$. 

A set system $M$ is a tuple $(E,\I)$ where $E$ is a set and $\I$ is any family of subsets of $E$. We call $M$ an \emph{independence system} if $\I$ is downward-closed ($Y \in \I, X \subseteq Y \Rightarrow X \in \I$) and $\emptyset \in \I$. The set $E$ is called the \emph{ground set} of $M$ (also denoted $E(M)$) and $\I$ is called the set of \emph{independent sets} of $M$ (also denoted ($\I(M)$). We call $M$ a \emph{matroid} if $M$ is an independence system and for all $X \subseteq E$, all inclusion-wise maximal independents subsets of $X$ have the same size.

The \emph{rank} $r_M(X)$ of a set $X \subseteq E$ is the maximum size of an independent subset of $X$. 
We omit the subscript $M$ when $M$ is clear from context. 
The \emph{rank} of $M$ itself is $r(E)$ and is also denoted $r(M)$. 
The maximal independent subsets of $E$ are called \emph{bases}.

The \emph{matroid polytope} of a matroid $M$ is the polytope $\mathcal{M} \subseteq \R^E$ such that $x \in \mathcal{M}$ if and only if 
\[
0 \leq \sum_{e \in X} x_e \leq r_M(X)\quad \forall X \subseteq E.
\]
It is a well-known theorem of Edmonds \cite{Edmonds1970} that $\mathcal{M} = \conv(\{\ind_I \mid I \in \I(M)\})$, see e.g. \cite[Theorem 13.12]{Korte2018}. 
When $(E,\mathcal{I},g)$ is an CBGT instance where $M=(E,\mathcal{I})$ is a matroid, the requirement that $g \in \conv(\{\ind_I \mid I \in \I\})$ is thus equivalent to
\begin{equation} 
    \label{eq:growth-and-rank}
    0 \leq \sum_{e\in X}g(e) \leq r(X)\quad \forall X \subseteq E.
\end{equation}

Let $M_1=(E_1,\mathcal{I}_1),\dots,M_k=(E_k,\mathcal{I}_k)$ be matroids on disjoint ground sets and let $E=E_1\cup\dots\cup E_k$. 
The \emph{direct sum} $M_1\oplus \dots\oplus M_k$ is the matroid $M=(E,\mathcal{I})$, where $\mathcal{I}=\{I\subseteq E\mid I\cap E_i\in\mathcal{I}_i\ \text{for $i\in[k]$}\}$. 

Let $G$ be a bipartite graph with $V(G)=S \cup T$. 
A subset $I\subseteq S$ is said to be \emph{matchable} if there is a matching in $G$ covering $I$. If $\mathcal{I}$ is the family of matchable sets, then $(S,\mathcal{I})$ is called \emph{transversal matroid}. 
The \emph{$k$-uniform matroid} with ground set $E$ and rank $k$ is defined by the independent sets $\{ I \subseteq E \mid |I| \leq k \}$. 
A \emph{partition matroid} is a direct sum of uniform matroids.
The \emph{graphic matroid} $(E,\I)$ has an associated (undirected) graph $G$ with edge set $E$, and  its family of independent sets is $\mathcal{I}= \{I\subseteq E \mid I \text{ is the edge set of a forest in } G\}$. A family of subsets  $\mathcal{L} \subseteq 2^{E}$ is \textit{laminar} if for any $L,L' \in \mathcal{L}$, we have $L \subseteq L'$, $L' \subseteq L$, or $L \cap L' = \emptyset$. 
Given a laminar family $\mathcal{L}$ and a function $b\colon \mathcal{L} \rightarrow \N$, a matroid $(E,\I)$ is \emph{laminar} if its set system is described as $\mathcal{I}= \{I\subseteq E \mid |I\cap L| \leq b(L) \text{ for each } L\in\mathcal{L} \}$.

\paragraph{Algorithms.} An \emph{algorithm} $\mathcal{A}$ for a family $\mathcal{F}$ of CBGT instances takes as input an instance $(E,\I,g) \in \mathcal{F}$ and outputs a schedule $\pi$ for $(E,\I,g)$ \textit{one step at a time}. 
We measure the efficiency of $\mathcal{A}$ based on the number of operations used to compute the set of bamboos to at each step. We say that an algorithm $\mathcal{A}$ \emph{implements} a schedule $\pi$ if $\pi$ is the sequence of sets obtained by running $\mathcal{A}$ indefinitely.

In \Cref{sec:FUN-algorithms}, where we present our polynomial-time algorithms for $k$-uniform, graphic, and laminar matroids, we assume that $\I$ is represented as just the integer $k$, the associated graph, and, respectively, the laminar family $\mathcal{L}$ and the upper bound $b(L)$ for each $L \in \mathcal{L}$. It is with respect to this input representation that our algorithms are polynomial-time. In \Cref{sec:general}, where we give an algorithm implementing a height-$O(\log n)$ schedule for arbitrary independence systems\footnote{In fact, as we argue, the algorithm can be used for arbitrary (not necessarily downward-closed) set systems.}, we assume that $\I$ is given implicitly via access to a \emph{maximum-weight independent-set oracle}. Given non-negative weights $w(e)$ for $e \in E$, such an oracle can find a set $I \in \I$ with maximum total weight $\sum_{e\in I}w(e)$. Note that for several natural classes of set systems (e.g., matroids, intersections of two matroids, and matchings in general graphs) polynomial-time implementations of such an oracle are well known. Furthermore, the oracle is trivial to implement given $\I$ explicitly as, e.g., in~\cite{cicerone2019FairHitting}.

\subsection{Approximation Factor of Low-Height Schedules}
\label{sec:approxfact}

The focus of this paper is on existence of schedules whose (absolute) height is low, and on algorithms for implementing such schedules. In this subsection, we show how one can turn such algorithms into \emph{approximation algorithms}, i.e., algorithms whose output is guaranteed to be at most a factor of $c$ away from the optimal schedule in terms of height.

In fact, all that is needed for an algorithm implementing a height-$c$ schedule to be a $c$-approximation algorithm is that $g$ is ``maximally'' scaled before it is given as input. For BGT, this observation reduces to the fact that a height-$c$ schedule is a $c$-approximation schedule when $g$ is such that $\sum_{e \in E} g(e) = 1$, which has been used in the literature.
In general, consider the primal--dual pair of LPs displayed in \Cref{fig:LPprimal}. It was observed in \cite[Section 3]{cicerone2019FairHitting} that $\Lambda(g)$, is a lower bound for the height $h(\pi^\star)$ of the optimal schedule. This gives the following, a proof of which we also give in the appendix using the terminology of this paper.

\begin{restatable}[$\star$]{thm}{optlb}
    \label{thm:optlb}
    Let $(E,\I,g)$ be a CBGT instance with $\Lambda(g) = 1$. The height $h(\pi^\star)$ of the optimal schedule $\pi^\star$ is at least $1$.
\end{restatable}

Note that, if $\Lambda(g) = 1$, then $g \in \conv(\{\ind_I \mid I \in \I\})$. Indeed, the vector $\lambda$ corresponding to the optimal primal solution is the set of coefficients of a convex combination $g' \in \conv(\{\ind_I \mid I \in \I\})$ (as $\sum_{I \in \I} \lambda(I) = 1$) and $g' \geq g$ (by the constraints of the LP). We can transform $\lambda$ such that all constraints (of the first type) are tight, without affecting $\sum_{I \in \I} \lambda(I)$, thus obtaining $g$ as a convex combination of $\{\ind_I \mid I \in \I\}$. Indeed, if there is an $e \in E$ with $g'(e) > g(e)$, then pick any $I \in \I$ containing $e$ with $\lambda(I) > 0$ and move some mass from $\lambda(I)$ to $\lambda(I \setminus \{e\})$. Repeat until $g' = g$.

To obtain a $c$-approximation schedule from an algorithm computing a height-$c$ schedule using \Cref{thm:optlb}, all we need to do is give $\hat{g} = g/\Lambda(g)$ to the algorithm instead of $g$. This guarantees that the obtained schedule $\pi$ is a $c$-approximation with respect to $\hat{g}$ (as $h(\pi)/h(\pi^\star) \leq c/1$ with respect to $\hat{g}$), and thus also with respect to $g$ as $h(\pi)/h(\pi^\star)$ is invariant under scaling of $g$.

Note that $\Lambda(g)$ can be computed in polynomial time with access to a maximum-weight independent-set oracle for $\I$. This follows directly from the fact that the separation problem for the dual (and thus the dual as a whole via the ellipsoid method) reduces to finding a maximum-weight independent set: Interpret a candidate solution $\{x_e\}_{e \in E}$ as weights and check if a maximum weight $I \in \I$ satisfies $\sum_{e \in I} x_e \leq 1$ using the oracle. If so, $x$ is feasible. If not, return the violated dual constraint for this $I$.

While we find the assumption that $\I$ is downward-closed to be natural, we also consider the case of non-downward-closed $\I$ in~\Cref{sec:general} as this is formally required for our results to be comparable with \cite{cicerone2019FairHitting}. An efficient $O(\log |E|)$-approximation algorithm can be obtained from our algorithm implementing a height-$O(\log |E|)$ schedule, which we present in~\Cref{sec:general}, even for non-downward-closed set systems.  For such set systems, it may not be that $g \in \conv(\{\ind_I \mid I \in \I\})$, even when $\Lambda(g) = 1$. However, the proof of \Cref{thm:optlb} does not assume that $\I$ is downward-closed, and the analysis of our height-$O(\log |E|)$ schedule only requires that there exists $g' \in \conv(\{\ind_I \mid I \in \I\})$ such that $g' \geq g$, which certainly does exist when $\Lambda(g) = 1$ as we argued above.

\begin{figure}
\centering

\begin{subfigure}{0.45\textwidth}
\centering
\[
\begin{aligned}
\Lambda(g) = \min \quad & \sum_{I\in\mathcal I} \lambda(I) \\
\text{s.t.} \quad
& \sum_{I\in\mathcal I: e \in I} \lambda(I) \geq g(e) && \forall e \in E, \\
& \lambda(I) \ge 0 && \forall I\in\mathcal I.
\end{aligned}
\]
\caption{Primal}
\label{fig:LPprimal-left}
\end{subfigure}
\hfill
\begin{subfigure}{0.45\textwidth}
\centering
\[
\begin{aligned}
\max \quad & \sum_{e\in E} x_e g(e) \\
\text{s.t.} \quad
& \sum_{e\in I} x_e \le 1 && \forall I\in\mathcal I, \\
& x_e \geq 0 && \forall e\in E.
\end{aligned}
\]
\caption{Dual}
\label{fig:LPprimal-right}
\end{subfigure}

\caption{Linear programs used to obtain approximation guarantees from height guarantees.}
\label{fig:LPprimal}
\end{figure}

\subsection{Combinatorial Pinwheel Scheduling}

An instance of \emph{Combinatorial Pinwheel Scheduling} (CPS) is a triple $(E,\I,a)$ where $(E,\I)$ is an arbitrary set system and $a \in \N^E$. 
We say that $(E,\I,a)$ is \emph{schedulable} if there exists a valid infinite schedule $\pi$ such that, for all $e \in E$, we have $e \in \bigcup_{i=t}^{t+a(e)-1} \pi(i)$ for all $t \in \N$.
The \emph{density} of $(E,\I,a)$ is the minimum $\rho$ such that there exist $\lambda(I)\geq 0$ for all $I\in\mathcal{I}$ with $\sum_{i\in\mathcal{I}}\lambda(I)=\rho$ and $\sum_{I\in\mathcal{I}\colon e\in I} \lambda(I)\geq 1/a(e)$ for all $e\in E$. 
The \emph{density threshold} $\rho^\star(\F)$ for a family $\F$ of CPS instances is the largest value $\rho$ such that any instance with density less than $\rho$ is schedulable. 

In this language, the recent confirmation of the pinwheel conjecture~\cite{Kawamura24} shows that $\rho^\star(\F) = 5/6$ when $\F = \{(E,\I,a) \mid \text{$(E,\I)$ is a $1$-uniform matroid}\}$. 
It follows from our results that $\rho^\star(\F) \geq 1/2$ for the much broader family $\F = \{(E,\I,a) \mid \text{$(E,\I)$ is a matroid}\}$. 
This is a consequence of the following straightforward observation.

\begin{restatable}{prop}{proppinwheel}
    \label{prop:pinwheel}
    Let $B = (E,\I,g)$ be a CBGT instance and let $P = (E,\I,a)$ be the CPS instance where $a(e) = \lfloor c/g(e) \rfloor$ for all $e \in E$ for some $c \in \R_{> 0}$. 
    Then $P$ is schedulable if and only if there exists an infinite valid schedule $\pi$ for $B$ with $h(\pi) \leq c$.
\end{restatable}
\begin{proof}
    Let $\pi$ be a schedule for $B$. 
    Then $h(\pi) \leq c$ iff $h^\pi(e) = \gamma_\pi(e)g(e) \leq c$ for all $e\in E$ iff $\gamma_\pi(e) \leq \lfloor  c/g(e)\rfloor = a(e)$ for all $e \in E$ iff $e$ occurs at least once in any period of $a(e)$ consecutive time steps in $\pi$ for all $e \in E$.
\end{proof}

This, together with \Cref{thm:matroidtwo} (\Cref{sec:matroid-height2}), implies a lower bound on the density threshold of $\{(E,\I,a) \mid \text{$(E,\I)$ is a matroid}\}$.

\begin{restatable}[$\star$]{corollary}{cordensity}
    Let $P = (E,\I,a)$ be a CPS instance where $(E,\I)$ is a matroid. 
    If $P$ has density at most $1/2$, then $P$ is schedulable.
\end{restatable}

Together with \Cref{thm:lower} (\Cref{ss:log-lower-bound}), \Cref{prop:pinwheel} also implies an upper bound on the density threshold of $\{(E,\I,a) \mid \text{$(E,\I)$ is an arbitrary set system}\}$. 

\begin{restatable}[$\star$]{corollary}{corddensityup}
    For any $n$, there exists a CPS instance $(E,\I,a)$ with $|E| \leq n$ and density $O(1/\log n)$ that is not schedulable.
\end{restatable}

\section{Existence of Height-2 Schedules on Matroids}
\label{sec:matroid-height2}

In this section, we show~\Cref{thm:matroidtwo}, which we restate here for convenience.

\matroidtwo

Throughout this section, fix a CBGT instance $(E,\I,g)$ where $M = (E,\I)$ is a matroid. We assume that no $e \in E$ has $g(e) = 0$. This is without loss of generality as such $e$ can simply be removed and the resulting submatroid considered instead. 
By \eqref{eq:growth-and-rank}, we can assume that
\begin{equation}\label{eq:full_rank}
    \sum_{e \in E} g(e) = r(M).
\end{equation}
Indeed, if this does not hold, we modify $g$ so that it does hold: 
Let $\lambda \colon \I(M) \rightarrow [0,1]$ be such that $\sum_{I \in \I(M)} \lambda(M) =1$ and $\sum_{I \in \I(M) : e \in I} \lambda(I) = g(e)$ for all $e \in E$. 
If $\lambda(I) = 0$ for all non-bases $I \in \I(M)$, then 
\[
    \sum_{e \in E} g(e) = \sum_{e \in E} \sum_{I \in \I(M) \colon e \in I} \lambda(I) = \sum_{I \in \I(M)} \lambda(I)|I| = r(M)\sum_{I \in \I(M)} \lambda(I) = r(M).
\]
If $\lambda(I) > 0$ for some non-basis $I \in \I(M)$, then we can simply pick basis $I' \supset I$, increase $\lambda(I')$ by $\lambda(I)$ and set $\lambda(I) = 0$, iteratively until~\eqref{eq:full_rank} holds.
Clearly, since we did not decrease $g(e)$ for any $e \in E$, it suffices to show that the modified instance has a schedule of height less than $2$.

In fact, we will show the stronger statement: $(E,\I,g)$ has a schedule $\pi$ with $d(\pi) < 1$.
We start by showing that this is indeed stronger.

\begin{restatable}{lemma}{disc1height}
    \label{lem:disc1height}
    If $\pi$ is a schedule with $d(\pi) < 1$, then $h(\pi) < 2$.
\end{restatable}
\begin{proof}
    Fix any $e \in E$. 
    As $d(\pi) < 1$, we have \[
        A(\pi,e,t) \in \{\lfloor g(e)t \rfloor,\lceil g(e)t\rceil\}
    \]
    for all $t \in [|\pi|]$. Let $1 \leq i < i'$. If $i'-i+1 \geq
    2/g(e)$, then \[
        A(\pi,e,i') \geq \left\lfloor g(e)\left(i+\frac{2}{g(e)}-1\right)
            \right\rfloor = \lfloor
        g(e)(i-1) \rfloor + 2 > \lceil g(e)(i-1) \rceil \geq A(\pi,e,i-1)
    \]
    and thus there must be some $j \in [i,i']$ such that $e \in \pi(j)$.
    Therefore, 
    \[
        h^\pi(e) = g(e)\gamma_\pi(e) = g(e)\sup_{1 \leq i < i'} \left\{i'-i+1 \;\middle\vert\; e \notin
        \cup_{j=i}^{i'} \pi(j)\right\} < g(e)\frac{2}{g(e)} = 2.
    \]
    As this holds for all $e \in E$, we have $h(\pi) = \max_{e \in E} h^\pi(e) < 2$, as desired.
\end{proof}
We remark that the converse of \Cref{lem:disc1height} does not hold. For example, for the BGT instance $(\{a,b\},g)$ with $g(a) = 0.9$ and $g(b) = 0.1$, the schedule $\pi = ababab...$ has $h(\pi) = 1.8$, but $d(\pi)$ is unbounded.

We continue by showing that, for an appropriately chosen $T$, a finite schedule $\pi$ with $|\pi| = T$ and $d(\pi) <1$ can be repeated indefinitely to obtain an infinite schedule $\pi^\infty$ with $d(\pi^\infty) < 1$.

\begin{restatable}{lemma}{finiteinfinite}
    \label{lem:finiteinfinite}
    Let $(E,\I,g)$ be a CBGT instance. 
    Let $T$ be such that $Tg(e)$ is an integer for all $e \in E$ and let $\pi$ be a schedule with $d(\pi) < 1$ and $|\pi| = T$.
    Then the infinite repetition $\pi^\infty$ of $\pi$ satisfies $d(\pi^\infty) < 1$. 
\end{restatable}
\begin{proof}
    As $d(\pi) < 1$ and $Tg(e)$ is an integer, we have
    \[
        |Tg(e)-A(\pi,e,T)| < 1 \Leftrightarrow Tg(e)-A(\pi,e,T) = 0.
    \]
    Thus $|T'g(e)-A(\pi,e,T')| = |(T' \bmod T)g(e)-A(\pi,e,T' \bmod T)|$
    for any $T' \in \N$, which implies $d(\pi^\infty) = d(\pi) < 1$.

\end{proof}

From now on, we fix the time horizon $T$ as the least common multiple of the denominators of $\{g(e) \mid e \in E\}$.
Note that for that choice of $T$, we have that $Tg(e)$ is an integer for all $e \in E$.

We aim to construct a finite schedule $\pi$ for $(E,\cI,g)$ with length $T$ and discrepancy $d(\pi)<1$.
Our approach identifies schedules of length $T$ with subsets of of the ground set $E\times [T]$.
Specifically, any subset $S\subseteq E \times [T]$ corresponds to a schedule $\pi$ where $e \in \pi(t)$ if and only if $(e,t) \in S$. 
Since an arbitrary subset $S\subseteq E\times [T]$ does not necessarily yield a valid schedule with low discrepancy, we enforce these constraints by defining two matroids, $M_T$ and $M_E$, over the common ground set $E\times [T]$.
We enforce validity by ensuring that the selection at every time step is an independent set in $M$; formally, we consider the the direct sum $M_T = \bigoplus_{t=1}^T M_t$ where each $M_t$ is a copy of $M$ on the ground set $E \times \{t\}$. 
The second matroid $M_E$ is constructed to enforce the requirement that $d(\pi)<1$.

The following lemma follows immediately from the definition of $M_T$.

\begin{restatable}{lemma}{lemMTvalid}
    \label{lem:lemMTvalid}
    Any basis $I \in \I(M_T)$ corresponds to a valid finite schedule $\pi$ (possibly with $d(\pi) \geq 1$).
\end{restatable}

The (more involved) definition of matroid $M_E$ is given in the next subsection. 
It will be defined such that it satisfies the following lemma, which is also proven in the next subsection.

\begin{restatable}{lemma}{basisdisc}
    \label{lem:basisdisc}
    Any basis $I \in \I(M_E)$ corresponds to a (possibly invalid) finite schedule $\pi$ with $d(\pi) < 1$.
\end{restatable}
In the last subsection of this section, we show the following.

\begin{restatable}{lemma}{common}
    \label{lem:common}
    $M_E$ and $M_T$ have a common basis.
\end{restatable}
Assuming \Cref{lem:basisdisc,lem:common}, we now give a proof of the main theorem of this section.

\begin{proof}[Proof of \Cref{thm:matroidtwo}.]
    By \Cref{lem:common}, let $I \in \I(M_E) \cap \I(M_T)$ be a common basis of $M_E$ and $M_T$ and let $\pi$ be the corresponding schedule. 
    \Cref{lem:lemMTvalid,lem:basisdisc} imply that $\pi$ is a valid schedule and that $d(\pi) < 1$. 
    By \Cref{lem:finiteinfinite}, $d(\pi^\infty) < 1$, where $\pi^\infty$ is the infinite repetition of $\pi$. 
    Finally, \Cref{lem:disc1height} allows us to conclude that $h(\pi^\infty) < 2$.
\end{proof}

\subsection{Construction of \texorpdfstring{$M_E$}{the discrepancy matroid} and Proof of \texorpdfstring{\Cref{lem:basisdisc}}{Lemma~\ref{lem:basisdisc}}}

For each $e \in E$, we define a matroid $M_e$ with ground set $\{e\} \times [T]$. 
This choice of taking $\{e\} \times [T]$ rather than just $[T]$ allow us to take a direct sum. 
We then define $M_E = \bigoplus_{e \in E} M_e$.

For the remainder of this subsection (except the proof of~\Cref{lem:basisdisc} at the end), let $e \in E$ be fixed. As $e$ is fixed, we may simply consider $[T]$ to be the ground set of $M_e$ instead of $\{e\} \times [T]$.

The intuition 
is as follows. We want a basis in $M_e$ to correspond to a schedule $\pi$ with $d^\pi(e) < 1$. 
If $d^\pi(e) < 1$, then the $i$-th cut of $e$ in $\pi$ must occur within a specific interval of time steps~$t$:
If the cut happens too early, then $A(\pi,e,t)$ exceeds $\lceil tg(e) \rceil$; if it happens too late, then $A(\pi,e,t)$ falls behind $\lfloor tg(e) \rfloor$.
We show that the interval $C_i$ as defined below is exactly the interval of time steps in which the $i$-th cut of $e$ must occur in order for $d^\pi(e) < 1$ to hold.

Define
\begin{equation}
    \label{eq:C} C_i = [\min\{t\in \N \mid \ce{tg(e)} = i\},\min\{t \in \N \mid
    \fl{tg(e)} = i\}] \cap \N
\end{equation}
for $i \in \N$. Note that $C_i$ is well defined as $0 < g(e) \leq 1$, and
thus $\fl{tg(e)} = i$ and $\ce{tg(e)} = i$ have integer solutions for $t$
for any $i \in \N$. 
Equivalently, we could define 
\[
    \label{eq:Cex} C_i = \left[\left\lfloor \frac{i-1}{g(e)}\right\rfloor+1, \left\lceil\frac{i}{g(e)}\right\rceil\right] \cap \N
\]
for all $i \in \N$, which is straightforward to derive from \eqref{eq:C}.

We define $G_e$ as the bipartite graph with $V(G_e) = [T] \cup \{c_1,c_2,\dots,c_{Tg(e)}\}$ and $N_{G_e}(c_i) = C_i$ for $i \in [Tg(e)]$. The graph $G_e$ for $T = 22$ and $g(e) = 4/11$ is illustrated in \Cref{fig:ge}. 
The set of independent sets $\I(M_e)$ of $M_e$ is \[
\I(M_e)=\{ I \subseteq [T] \mid \text{$I$ is covered by some matching in $G_e$}\}.
\]
It is well-known that $M_e = ([T],\I(M_e))$ is a (transversal) matroid. To show that the bases of $M_e$ correspond to schedules $\pi$ with $d^\pi(e) < 1$, we first need an auxiliary lemma about the $C_i$.

\begin{figure}
\centering
\begin{tikzpicture}[
  every node/.style={font=\scriptsize},
  tnode/.style={circle,draw,minimum size=5.5mm,inner sep=0pt},
  cnode/.style={circle,draw,minimum size=5.5mm,inner sep=0pt},
  edge/.style={draw, line width=0.3pt, opacity=1}
]
\node[tnode] (t1) at (0.000,1.900) {$ 1 $};
\node[tnode] (t2) at (0.650,1.900) {$ 2 $};
\node[tnode] (t3) at (1.300,1.900) {$ 3 $};
s\node[tnode] (t4) at (1.950,1.900) {$ 4 $};
\node[tnode] (t5) at (2.600,1.900) {$ 5 $};
\node[tnode] (t6) at (3.250,1.900) {$ 6 $};
\node[tnode] (t7) at (3.900,1.900) {$ 7 $};
\node[tnode] (t8) at (4.550,1.900) {$ 8 $};
\node[tnode] (t9) at (5.200,1.900) {$ 9 $};
\node[tnode] (t10) at (5.850,1.900) {$ 10 $};
\node[tnode] (t11) at (6.500,1.900) {$ 11 $};
\node[tnode] (t12) at (7.150,1.900) {$ 12 $};
\node[tnode] (t13) at (7.800,1.900) {$ 13 $};
\node[tnode] (t14) at (8.450,1.900) {$ 14 $};
\node[tnode] (t15) at (9.100,1.900) {$ 15 $};
\node[tnode] (t16) at (9.750,1.900) {$ 16 $};
\node[tnode] (t17) at (10.400,1.900) {$ 17 $};
\node[tnode] (t18) at (11.050,1.900) {$ 18 $};
\node[tnode] (t19) at (11.700,1.900) {$ 19 $};
\node[tnode] (t20) at (12.350,1.900) {$ 20 $};
\node[tnode] (t21) at (13.000,1.900) {$ 21 $};
\node[tnode] (t22) at (13.650,1.900) {$ 22 $};

\node[cnode] (c1) at (0.650,0.000) {$ c_{1} $};
\node[cnode] (c2) at (2.275,0.000) {$ c_{2} $};
\node[cnode] (c3) at (4.225,0.000) {$ c_{3} $};
\node[cnode] (c4) at (5.850,0.000) {$ c_{4} $};
\node[cnode] (c5) at (7.800,0.000) {$ c_{5} $};
\node[cnode] (c6) at (9.425,0.000) {$ c_{6} $};
\node[cnode] (c7) at (11.375,0.000) {$ c_{7} $};
\node[cnode] (c8) at (13.000,0.000) {$ c_{8} $};

\draw[edge] (c1) -- (t1);
\draw[edge] (c1) -- (t2);
\draw[edge] (c1) -- (t3);
\draw[edge] (c2) -- (t3);
\draw[edge] (c2) -- (t4);
\draw[edge] (c2) -- (t5);
\draw[edge] (c2) -- (t6);
\draw[edge] (c3) -- (t6);
\draw[edge] (c3) -- (t7);
\draw[edge] (c3) -- (t8);
\draw[edge] (c3) -- (t9);
\draw[edge] (c4) -- (t9);
\draw[edge] (c4) -- (t10);
\draw[edge] (c4) -- (t11);
\draw[edge] (c5) -- (t12);
\draw[edge] (c5) -- (t13);
\draw[edge] (c5) -- (t14);
\draw[edge] (c6) -- (t14);
\draw[edge] (c6) -- (t15);
\draw[edge] (c6) -- (t16);
\draw[edge] (c6) -- (t17);
\draw[edge] (c7) -- (t17);
\draw[edge] (c7) -- (t18);
\draw[edge] (c7) -- (t19);
\draw[edge] (c7) -- (t20);
\draw[edge] (c8) -- (t20);
\draw[edge] (c8) -- (t21);
\draw[edge] (c8) -- (t22);
\end{tikzpicture}
\caption{$G_e$ for $g(e) = 4/11$ and $T=22$}
\label{fig:ge}
\end{figure}
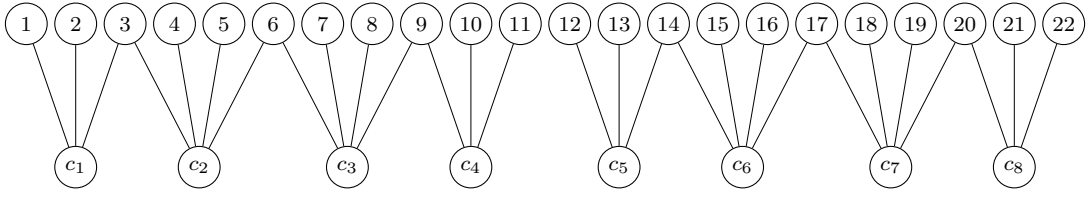

\begin{restatable}{lemma}{increasing}
    \label{lem:increasing}
    If $i \in \N$ is such that $i/g(e)$ is not an integer, then \[
        \min C_i < \max C_i = \min C_{i+1},
    \]
    and if $i/g(e)$ is an integer, then $\max C_i + 1 = \min C_{i+1}$.
\end{restatable}

\begin{proof}
    Let $i \in \N$. If $i/g(e)$ is an not an integer, then, as $g(e) \leq 1$, \[
        \min C_i = \left\lfloor \frac{i-1}{g(e)} \right\rfloor + 1 \leq \left\lfloor \frac{i-1}{g(e)} + \frac{1}{g(e)}-1 \right\rfloor + 1 = \left\lfloor \frac{i}{g(e)} \right\rfloor < \left\lceil \frac{i}{g(e)} \right\rceil = \max C_i
    \]
    and \[
        \max C_i = \left\lceil \frac{i}{g(e)} \right\rceil = \left\lfloor \frac{i+1-1}{g(e)} \right\rfloor+1 = \min C_{i+1}.
    \]
    Now, if $i/g(e)$ is an integer, then \[
        \max C_i +1= \frac{i}{g(e)} +1 = \frac{i+1-1}{g(e)} +1= \min C_{i+1},
    \]
    which completes the proof.
\end{proof}

We next show two lemmas about the matroid $M_e$.

\begin{restatable}{lemma}{lemrank}
    \label{lem:rank}
    $r(M_e) = Tg(e)$.
\end{restatable}

\begin{proof}
    Let $I \in \I(M_e)$. Note that $|I| > Tg(e)$ is not possible
    as $I$ must be covered by some matching in $G_e$, and there are only $Tg(e)$ vertices
    $c_1,c_2,\dots,c_{Tg(e)}$ to match $I$ to. Furthermore, $\{\{\min C_I,c_i\} \mid i \in [Tg(e)]\}$ is a matching in $G_e$ of size $Tg(e)$ as $\min C_i$ and $\min C_j$ are distinct for distinct $i,j$ by \Cref{lem:increasing}.
\end{proof}

\begin{restatable}{lemma}{lemMe}
    \label{lem:Me}
     If $I \in \I(M_e)$ is a basis and $\pi$ is any schedule such that $e \in \pi(t)$ exactly when $t \in I$, then $d^\pi(e) < 1$.
\end{restatable}

\begin{proof}
    Let $k = Tg(e)$ and let $I \in \I(M_e)$ be a basis of $M_e$, which by \Cref{lem:rank} satisfies $|I| = k$. Let $Y$ be a matching in $G_e$ covering $I$, and let 
    $\pi$ be any schedule such that $e \in \pi(t)$ exactly when $t \in I$.

    To show that $d^\pi(e) < 1$, we must show that
    \[
        \fl{tg(e)} \leq A(\pi,e,t) \leq \ce{tg(e)}.
    \]
    for all $t \in [T]$. To this end, fix any $t \in [T]$.

    Let $I = \{t_1,t_2,\dots,t_k\}$ be such that $t_1 < t_2 < \dots < t_k$.
    For every $i \in [k]$, there must be some $j \in [k]$ such that $t_jc_i \in Y$ as $|I| = k$. We claim that $j=i$. Indeed, if this is not the case, then there is  pair $i,j \in [k]$ with $i < j$ such that $t_jc_i \in Y$. But by \Cref{lem:increasing}, the intervals $C_{i+1},C_{i+2},\dots,C_{k}$
     do not contain any element less than $t_j$. This implies that
    $t_1,t_2,\dots,t_j$ must be matched to $c_1,c_2,\dots,c_{i}$, which is
    impossible for $i < j$. Hence, $t_ic_i\in Y$ holds for all $i\in [k]$.
    
    We now show $A(\pi,e,t) \geq \fl{tg(e)}$. Let $i = \fl{tg(e)}$. As $t_ic_i \in Y$, we have $t_i \in N(c_i)=C_i$. Thus,
    the $i$-th
    occurrence of $e$ in $\pi$ is at index $t_i \leq \max C_i = \min\{t' \mid
    \fl{t'g(e)} = i\} \leq t$, and thus $A(\pi,e,t) \geq i = \fl{tg(e)}$.

    We now show $A(\pi,e,t) \leq \ce{tg(e)}$. Let $i = \ce{tg(e)}$. If $i = k$, then we are done as
    $\pi$ contains $k$ occurrences of $e$ in total. 
    Thus, we assume $i \leq k-1$. 
    As $t_{i+1}c_{i+1} \in Y$, we have that $t_{i+1} \in N(c_{i+1})=C_{i+1}$. 
    Therefore, the $(i+1)$-th occurrence of $e$ in $\pi$ is at index $t_{i+1} \geq \min C_{i+1} = \min\{t' \mid \ce{t'g(e)} = i+1\} > t$, and thus $A(\pi,e,t) < i+1$, which is equivalent to $A(\pi,e,t) \leq i = \ce{tg(e)}$.
\end{proof}

We now give a proof of \Cref{lem:basisdisc}.

\begin{proof}[Proof of \Cref{lem:basisdisc}]
    Let $I \in \I(M_E)$ be a basis of $M_E$ and let $\pi$ be the corresponding schedule. 
    As $M_E = \bigoplus_{e \in E} M_e$, the set $I \cap E(M_e)$ is a basis in $M_e$ for all $e \in E$. 
    As $e \in \pi(t)$ exactly when $(e,t) \in I \cap E(M_e)$, we have $d^\pi(e) < 1$ for all $e \in E$ by \Cref{lem:Me}. 
    Thus, $d(\pi) = \max_{e \in E} d^\pi(e) < 1$.
\end{proof}

\subsection{Proof of~\texorpdfstring{\Cref{lem:common}}{Lemma~\ref{lem:common}}}

We start by observing that $M_E$ and $M_T$ have the same rank as 
\[
r(M_E) = \sum_{e \in E} r(M_e) = \sum_{e \in E} Tg(e) = T \cdot r(M) = \sum_{t=1}^T r(M_t) = r(M_T),
\]
where we used the definition of $M_E$ and $M_T$ as direct sums for the first and last equalities, \Cref{lem:rank} for the second one, and our assumption that $\sum_{e \in E} g(e) = r(M)$ for the third.

Now, consider the \textit{matroid intersection polytope} $\mathcal{P} \subset \R^{E \times [T]}$ of $M_E$ and $M_T$ where $y \in \mathcal{P}$ if and only if
\begin{alignat}{2}
    \label{eq:first} 0 \leq \sum_{(e,t) \in X} y_{e,t} &\leq r_{M_T}(X) \ \ \ \ &&\forall X \subseteq E
    \times [T] \\
    \label{eq:second} 0 \leq \sum_{(e,t) \in X} y_{e,t} &\leq r_{M_E}(X) &&\forall X
    \subseteq E \times [T].
\end{alignat}
It is a well-known theorem of Edmonds \cite{Edmonds1970} that $\mathcal{P}$ has integer vertices. 
Let $y^\star$ be such that \[
    y^\star_{e,t}= g(e) 
\] for all $e \in E,t \in [T]$.
We show that $y^\star \in \mathcal{P}$.

\begin{restatable}{lemma}{yfirst}
\label{lem:yfirst}
$y^\star$ satisfies \eqref{eq:first}.
\end{restatable}

\begin{proof}
    For each $t \in [T]$, the set of points $y \in \R^{E \times \{t\}}$ such that
    \[
        0 \leq \sum_{e \in Y} y_{e,t} \leq r_{M_t}(Y \times \{t\})\ \ \ \forall Y \subseteq E
    \]
    is exactly $\conv(\{\ind_I \mid I \in \I(M_t)\})$ as $M_t$ is a matroid. Thus, the above holds for
    $y^\star_{\cdot,t} = g$ for any $t \in [T]$ by the assumption that $g \in \conv(\{\ind_I \mid I \in \I\})$. Thus, for any $X \subseteq E \times [T]$ we have that \[
        0 \leq \sum_{(e,t) \in X} y^\star_{e,t} = \sum_{t = 1}^T \sum_{(e,t) \in X} y^\star_{e,t} \leq \sum_{t =1}^T r_{M_t}(X \cap E(M_t)) = r_{M_T}(X).
    \]This completes the proof.
\end{proof}

\begin{restatable}{lemma}{ysecond}
    \label{lem:fractionalsolution}
    $y^\star$ satisfies \eqref{eq:second}.
\end{restatable}

\begin{proof}
    We show that $y^\star_{e,\cdot} = \textbf{1}g(e)$ is a convex combination of $\{\ind_I \mid \I(M_e)\}$ for any $e \in E$. If this holds, then
    as $M_e$ is a matroid, it follows that $y^\star$ satisfies \[
        0 \leq \sum_{t \in Z} y_{e,t}^\star \leq r_{M_e}(\{e\} \times Z)\ \ \ \forall Z \subseteq [T]
    \]
    for all $e \in E$, from which we get \[
         0 \leq \sum_{(e,t) \in X} y^\star_{e,t} = \sum_{e' \in E} \sum_{(e',t) \in X} y^\star_{e',t} \leq \sum_{e \in E} r_{M_e}(X \cap E(M_e)) = r_{M_E}(X).
     \]

    Fix any $e \in E$. We now proceed to show that $y^\star_{e,\cdot} = \textbf{1}g(e) \in \conv(\{\ind_I \mid \I(M_e)\})$
    by showing that there exists a $\lambda : \I(M_e) \rightarrow [0,1]$ with $\sum_{I \in \I(M_e)} \lambda(I) = 1$ and $\sum_{I \in \I(M_e)} \lambda(I) \cdot \ind_I = \textbf{1}g(e)$.
    
    Towards this, consider the \emph{fractional matching polytope} $\mathcal{P}' \subset \R^{E(G_e)}$ of $G_e$ (the graph through which we defined $M_e$) such that $x \in \mathcal{P}'$ iff
    \begin{alignat*}{2}
        \sum_{uv \in E(G_e)} x_{uv} &\leq 1\ \ \ &&\forall v \in V(G_e) \\
        x_{uv} &\geq 0 &&\forall uv \in E(G_e).
    \end{alignat*}
     Define \[
    x_{c_it} = \begin{cases}
         tg(e)-i+1 & \text{if } t = \min C_i \\
         i-(t-1)g(e) & \text{if } t = \max C_i \\
         g(e) & \text{otherwise}
    \end{cases}
    \]
    for all $c_it \in E(G_e)$.

    \begin{restatable}{claim}{claimmatching}
        \label{claim:matching} We have $
            \sum_{c_it \in E(G_e)} x_{c_it} = g(e)
        $
        for all $t \in [T]$, and $x \in \mathcal{P}'$.
    \end{restatable}
    \begin{proof}[Proof of \Cref{claim:matching}]
    \renewcommand{\qedsymbol}{$\diamond$}
    Let $t \in [T]$ be arbitrary. We consider $4$ simple cases.

    \smallskip
    Case 1: $\min C_i < t < \max C_i$ for some $i \in [Tg(e)]$.
    In this case \[
        \sum_{ut \in E(G_e)} x_{ut} = x_{c_it} = g(e).
    \]

    \smallskip
    Case 2: $t = \max C_i = \min C_{i+1}$ for some $i \in [Tg(e)-1]$.
    In this case \[
        \sum_{ut \in E(G_e)} x_{ut} = x_{c_it}+x_{c_{i+1}t} = i-(t-1)g(e)+tg(e)-(i+1)+1 = g(e).
    \]

    \smallskip
    Case 3: $t = \max C_i = \min C_{i+1}-1$ for some $i \in [Tg(e)-1]$ or $t=T$. In this case, we have $i/g(e) = t$ by \Cref{lem:increasing} (or as $t=T$) and so \[
        \sum_{ut \in E(G_e)} x_{ut} = x_{c_it} = i-(t-1)g(e) = i-(i/g(e)-1)g(e) = g(e)
    \]

    \smallskip
    Case 4: $t = \min C_i = \max C_{i-1}+1$ for some $i \in [Tg(e)] \setminus \{1\}$ or $t=1$. In this case, we have $(i-1)/g(e) = t-1$ by \Cref{lem:increasing} (or as $t=1$) and so \[
        \sum_{ut \in E(G_e)} x_{ut} = x_{c_it} = x_{c_it} = tg(e)-i+1 = ((i-1)/g(e)+1)g(e)-i+1 = g(e).
    \]

    \smallskip

    Lastly, observe that $x \in \mathcal{P}'$, as for any $i \in [Tg(e)]$, we have \[
        \sum_{c_it \in E(G_e)} x_{c_it} = (\min C_i)g(e)-i+1 + (\max C_i - \min C_i -1)g(e) + i-(\max C_i -1)g(e) = 1.
    \]
    This completes the proof of the claim.
    \end{proof}

    Let $\M$ denote the set of all matchings in $G_e$. It is well-known that $\mathcal{P}'$ has integer vertices as $G_e$ is bipartite (e.g., \cite[Theorem 11.4]{Korte2018}). Thus, every vertex of $\mathcal{P}'$ is the incidence vector $\ind_A$ of some $A \in \M$. 
    By this and as $x \in \mathcal{P}'$ (by \Cref{claim:matching}), $x$ is a convex combination of $\{\ind_A \mid A \in \M\}$. Let $\lambda': \M \rightarrow [0,1]$ with $\sum_{A \in \M} \lambda'(M) = 1$ such that \begin{equation}
        \label{eq:xge} \sum_{A \in \M} \lambda'(A) \cdot \ind_A = x.
    \end{equation}
    
    Now, for all $A \in \M$, define $T(A) = \{t \in [T] \mid \text{$t$ covered by $A$}\}$ as the set of endpoints of $A$ in $[T]$. Define
    $\lambda: \I(M_e) \rightarrow [0,1]$ by setting  \[
        \lambda(I) = \sum_{A \in \M\colon T(A) = I} \lambda'(A)
    \]
    for all $I \in \I(M_e)$.
    Then,
    \[
        \sum_{I \in \I(M_e)} \lambda(I)\ind_{I} = \sum_{I \in \I(M_e)} \left(\sum_{A \in \M: T(A) = I} \lambda'(A)\right)\ind_{I} = \sum_{A \in \M} \lambda'(A)\ind_{T(A)} = g(e)\textbf{1},
    \]
    where the first equality follow from the definition of $\lambda(I)$, the second from the fact $T(A) \in \I(M_e)$ for any matching $A \in \M$, and the last by \Cref{claim:matching} and \eqref{eq:xge}. Thus, we have expressed $g(e)\textbf{1}^T$ as a convex combination of independent sets of $M_e$, which is what we wanted.
 \end{proof}

We are now ready to prove \Cref{lem:common}.

\begin{proof}[Proof of \Cref{lem:common}]
    It follows from \Cref{lem:yfirst} and \Cref{lem:fractionalsolution} that any $y \in
    \mathcal{P}$ which maximizes $\sum_{e \in E, t \in [T]} y_{e,t}$ has a sum of entries at
    least $\sum_{(e,t) \in E \times [T]} y^\star_{e,t} = T \cdot r(M) = r(M_T) = r(M_E)$. As $\mathcal{P}$ has integer vertices, which are binary by definition of $\mathcal{P}$, there exists
    a binary vector $y_{\text{bin}}\in\mathcal{P}$ with at least $r(M) \cdot T$ entries that are $1$. By the definition of $\mathcal{P}$ and the fact that $r(M_E) = r(M_T) =
    r(M) \cdot T$, $y_{\text{bin}}$ is the incidence vector of a common basis of $M_E$ and $M_T$.
\end{proof}

\section{Efficient Implementations of Low-Height Schedules on Special Matroids}
\label{sec:FUN-algorithms}
In this section, we give polynomial-time algorithms for special classes of CBGT instances on matroids.

In \Cref{subsec:uniformFUN}, we give a polynomial-time algorithm which implements a schedule $\pi$ with $h(\pi) < 2$ for any CBGT instance $(E,\I,g)$ where $(E,\I)$ is a uniform matroid. Here, we assume that $\I$ is represented simply by the integer $k$ such that $\I = \{X \subseteq E \mid |X| \leq k\}$. As BGT is equivalent to CBGT for $1$-uniform matroids, we obtain a polynomial-time algorithm for BGT with the same height guarantee as a corollary.

In \Cref{subsec:Graphic_laminar}, we give polynomial-time algorithms which implement height-$4$ schedules for graphic matroids and laminar matroids. The algorithms use our algorithm for $1$-uniform matroids as a subroutine. Although \Cref{thm:matroidtwo} guarantees the existence of a schedule with height at most~$2$ for all matroids, we deliberately sacrifice height in favor of a more efficient implementation.

\subsection{Height-2 Schedules for Uniform matroids}
\label{subsec:uniformFUN}
Our algorithm is based on so-called \emph{tree schedules} introduced by Bar-Noy, Nisgav, and Patt-Shamir~\cite{barNoy2001Nearly} and later used, for example, in~\cite{bar-noyEfficientPeriodicScheduling2002}. We note that for $1$-uniform matroids, our algorithm essentially reduces to the algorithm given in \cite[Algorithm \texttt{binMax}]{bar-noyEfficientPeriodicScheduling2002}. It was, however, not shown in \cite{bar-noyEfficientPeriodicScheduling2002} that the implemented schedule has height below $2$.

We now define what we call the \textbf{Fuse--Unfuse} schedule.
Given two bamboos $e$ and $f$ in $E$, to \emph{fuse} $e$ and $f$ means to replace them with a new bamboo $e \circ f$ whose growth rate is $g(e \circ f) =
2\max\{g(e),g(f)\}$. 
Given a schedule $\pi$ containing $e \circ f$, to \emph{unfuse} $e \circ f$ in $\pi$ means to replace the occurrences of $e \circ f$ in $\pi$ alternately by $e$ and $f$.
For simplicity, we refer to the bamboo with minimum growth rate as the
\emph{slowest} bamboo in $E$, by breaking ties arbitrarily.

\begin{restatable}{defn}{FUN}
    \label{def:FUN} The \emph{Fuse--Unfuse schedule} $\pi$ of a BGT instance $(E,g)$ is
    defined inductively. If $E = \{e\}$, then $\pi$ is the infinite repetition of $e$. If $|E| \geq 2$, fuse the two slowest bamboos $e,f \in E$ to obtain $E'$ and let $\pi$ be the Fuse--Unfuse schedule of $E'$ with $e \circ 
    f$ unfused.
\end{restatable}

We state an important property for fusing-unfusing bamboos.
\begin{restatable}{lemma}{superbamboo}
    \label{lem:superbamboo}
    Let $\pi$ be a schedule, let $e \circ f \in \pi$, and let $\pi'$ be
    obtained from $\pi$ by unfusing $e$ and $f$. Then $h(\pi') \leq
    h(\pi)$.
\end{restatable}

\begin{proof}
    The maximum height reached by any bamboo which is not $e \circ f$ is
    unaffected by the unfusion of $e \circ f$. Thus, it suffices to show that
    the maximum height reached by $e$ or $f$ in $\pi'$ is at most the
    maximum height of $e \circ f$ in $\pi$. 

    The maximum height reached by $e \circ f$ in $\pi$ is $\gamma_\pi(e \circ f)
    \cdot g(e \circ f)$. Since $e \circ f$ is replaced by $e$ and $f$ in
    alternation, the recurrence time of both $e$ and $f$ in $\pi'$ is at
    most $2\gamma_\pi(e \circ f)$. Thus, the maximum height reached by $e$ or $f$ in
    $\pi'$ is at most \[
        \max\{\gamma_{\pi'}(e) \cdot g(e), \gamma_{\pi'}(f) \cdot
        g(f)\} \leq 2\gamma_\pi(e \circ f) \cdot \max\{g(e),g(f)\} = \gamma_\pi(e\circ
        f)\cdot g(e \circ f) = h^\pi(e \circ f),
    \]
    which completes the proof.
\end{proof}

We are ready to give the algorithm.
The high-level idea of \Cref{alg:fun} is as follows. The algorithm builds $k$ binary trees $T_1, \dots, T_k$ arising from the recursive fusion of bamboos in $E$. The tree $T_i$ can then be used to output the $i$-th bamboo to cut on each day. Hence, the algorithm returns a periodic schedule. See \Cref{example:FUN-uniform}.

\begin{algorithm}

\caption{Fuse--Unfuse Scheduling Algorithm.}
\label{alg:fun}
 \hspace*{\algorithmicindent} \textbf{Input:} A CBGT instance $(E,\I,g)$ where $(E,\I)$ is a $k$-uniform matroid,  $|E| \geq k$ and $g(e) > 0$ for all $e \in E$. \\
 \hspace*{\algorithmicindent} \textbf{Output:} $\pi = \pi(1),\pi(2),\ldots$. 
\begin{algorithmic}[1]
    \Statex \textbf{Preprocessing:}
\State Initialize a min-heap $H$ over $E$ keyed by growth rate $g$.
\While{$|H| > k$}
    \State Extract the two elements $T_L,T_R$ of minimum growth rate from $H$.
    \State Create a new tree $T$ with children $T_L$ and $T_R$, and assign $g(T) = 2 \cdot \max\{g(T_L), g(T_R)\}$.
    \State Insert $T$ into $H$.
\EndWhile
\State Let $T_1,T_2,\dots,T_k$ be the remaining trees in $H$.
\State Initialize a status bit $\sigma(v) \gets 0$ for every internal vertex $v \in V(T_i)$ and $i \in [k]$.

\Statex
\Statex \textbf{Execution (at time $t$):}
\State $I \gets \emptyset$.
\For{$i=1,2,\dots,k$}
    \State $v \gets$ root of $T_i$.
    \While{$v$ is not a leaf}
        \State $\sigma(v) \gets 1 - \sigma(v)$.
            \State If $\sigma(v) = 0$, then $v \gets$ left child of $v$. Otherwise  $v \gets$ right child of $v$.
    \EndWhile
    \State $I \gets I \cup \{v\}$
\EndFor
\State Cut $I$.
\end{algorithmic}
\end{algorithm}

\begin{example}
\label{example:FUN-uniform}
    Consider the instance of a $2$-uniform matroid $M$ on ground set $\{a,b,c,d,e\}$ with  growth rates $g(a)=0.1, g(b)=0.2, g(c)=0.5,g(d)=0.5, g(e)=0.3$.
    The 2 trees generated by \Cref{alg:fun} are shown in \Cref{fig:FUNuniform}, from which we obtain the schedules $\pi_1=(b,e,a,e,\dots)$ and $\pi_2=(d,c,\dots)$.
    The schedule for $M$ given by \Cref{alg:fun} is $\pi=(\{b,d\},\{e,c\},\{a,d\}, \{e,c\} ,\dots)$.
\end{example}
\begin{figure}
\centering
\begin{tikzpicture}[
  every node/.style={font=\footnotesize},
  tnode/.style={circle,draw,minimum size=2em,inner sep=0pt},
  cnode/.style={circle,draw,minimum size=2em,inner sep=0pt},
  edge/.style={draw, line width=1.2pt}
]
\node[tnode] (Teab) at (0,0) {$T_{eab}$};
\node[below = 0.2em of Teab] {\textcolor{gray}{0.8}};
\node[tnode, below left = 2em and 2em of Teab] (e) {$ e $};
\node[below = 0.2em of e] {\textcolor{gray}{0.3}};
\node[tnode, below right = 2em and 2em of Teab] (Tab) {$ T_{ab}$};
\node[below = 0.2em of Tab] {\textcolor{gray}{0.4}};
\node[tnode, below left = 2em and 2em of Tab] (a) {$ a $};
\node[below = 0.2em of a] {\textcolor{gray}{0.1}};
\node[tnode, below right = 2em and 2em of Tab] (b) {$b$};
\node[below = 0.2em of b] {\textcolor{gray}{0.2}};

\draw[edge] (Teab) -- (e);
\draw[edge] (Teab) -- (Tab);
\draw[edge] (Tab) -- (a);
\draw[edge] (Tab) -- (b);

\node[tnode] (Tcd) at (7,0) {$T_{cd}$};
\node[below = 0.2em of Tcd] {\textcolor{gray}{1}};
\node[tnode, below left = 2em and 2em of Tcd] (c) {$ c $};
\node[below = 0.2em of c] {\textcolor{gray}{0.5}};
\node[tnode, below right = 2em and 2em of Tcd] (d) {$ d$};
\node[below = 0.2em of d] {\textcolor{gray}{0.5}};

\draw[edge] (Tcd) -- (c);
\draw[edge] (Tcd) -- (d);
\end{tikzpicture}
\caption{Binary trees obtained in the preprocessing phase of \Cref{alg:fun} for the 2-uniform matroid with ground set $\{a,b,c,d,e\}$.}
\label{fig:FUNuniform}
\end{figure}
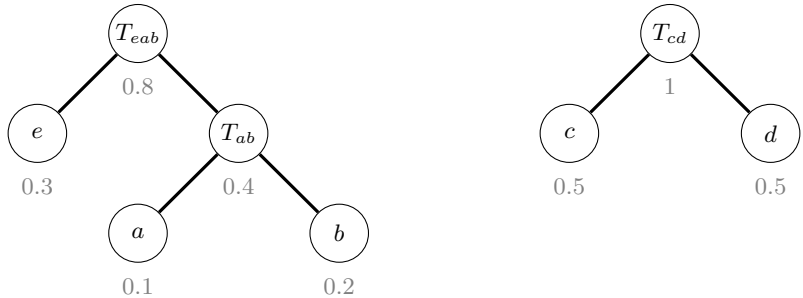

\begin{restatable}{thm}{kuniform}
    \label{lem:kuniform}
    \Cref{alg:fun} implements a height-$2$ schedule for any CBGT instance $(E,\I,g)$, where $(E,\I)$ is a $k$-uniform matroid in polynomial-time.
\end{restatable}

\begin{proof}
    First of all, note that the assumption that $|E| \geq k$ and that $g(e) > 0$ for all $e \in E$ is only for convenience as if $k > |E|$, then we can just run our algorithm with $k = |E|$ instead and bamboos $e \in E$ with $g(e) = 0$ can be removed. 
    
    We start by observing that we can view $\pi$ as $k$ schedules $\pi_1,\pi_2,\dots,\pi_k$ executed in parallel, one for each tree $T_1,T_2,\dots,T_k$. 
    Note that at each day, for each $i \in [k]$ the algorithm follows a unique root-to-leaf path in $T_i$ and cuts exactly one leaf of $T_i$. The bit-flipping rule guarantees that the traversal is well defined for all days in each tree. Hence, the algorithm returns a valid schedule $\pi$ as only $k$ bamboos are cut every day. 
    Furthermore, for each $i \in [k]$, $\pi_i$ is exactly the Fuse--Unfuse schedule for the leaves of $T_i$
    
    We show that $h(\pi) < 2$. By the above argument, we have to show that $g(T_i) < 2$ for each $i \in [k]$ and $h(\pi) = \max_{ i \in [k]} h(\pi_i) < 2$ follows from \Cref{lem:superbamboo}.

    Towards this, let $m$ be the number of times that the While-loop on lines $2$ to $6$ is executed. We henceforth refer to this as simply ``the loop''. Note that the claim is trivial if $m=0$ as $g(e) \leq 1$ for all $e \in E$ by \Cref{eq:growth-and-rank}. Denote by $H^0,H^1,H^2,\dots,H^m$ the set of trees in $H$ after the $i$-th iteration of the loop. This way, $H^m = \{T_1,T_2,\dots,T_k\}$. Also denote by $T_L^i$, $T_R^i$, and $T^i$ the trees $T_L$, $T_R$ and, $T$, respectively, in the $i$-the iteration of the loop for $i \in [m]$. We assume without loss of generality that $g(T_R^i) \geq g(T_L^i)$ such that $g(T^i) = 2g(T^i_R)$ for all $i \in [m]$. Define $g(H^i) = \sum_{T \in H^i} g(T)$ for $i \in [m]$ and note that $g(H^0) = \sum_{e \in E} g(e) \leq k$ by \eqref{eq:growth-and-rank}.

    Consider some iteration $i \in [m]$ of the loop. We create $H^i$ from $H^{i-1}$ by removing $T_L^i$ and $T_R^i$ and adding the new tree $T^i$ with $g(T^i) = 2g(T_R^i)$. Therefore, $g(H^{i}) = g(H^{i-1}) + g(T_R^i)-g(T_L^i)$. Now, notice that the tree $T_L^{i+1}$ removed in iteration $i+1$ has $g(T^{i+1}_L) \geq g(T^{i}_R)$ as $T_L$ and $T_R$ are chosen to be the slowest bamboos in $H^{i-1}$ in each iteration $i$.
    Therefore \begin{equation}
        \label{eq:contradictme}
        g(H^m) = g(H^0) + \sum_{i=1}^m (g(T_R^i)-g(T_L^i)) \leq k + g(T_R^m)-g(T_L^1) < k + g(T_R^m).
    \end{equation}

    We now notice that $g(T^m) \geq g(T^{m-1}) \geq \dots \geq g(T^1)$. This follows from the fact that $g(T^i) = 2g(T_R^i)$ and that the growth rate of the second-smallest element $T_R^i$ in $H_i$ only increases with $i$.
    
    We claim that $g(T^m) < 2$. Indeed, if this holds, then no tree with growth rate $2$ or more is ever added to $H$; in particular, $g(T_j) < 2$ for all $j \in [k]$. Assume for the sake of contradiction that $g(T^m) = 2g(T_R^m) \geq 2$. By renumbering the trees $T_1,T_2,\dots,T_k$, we may assume without loss of generality that $T^m = T_k$. As the $m$-th iteration is the last, we must have $H^{m-1} = \{T_1,T_2,\dots,T_{k-1},T_R^i,T_L^i\}$. As $T_R^m$ and $T_L^m$ are the two slowest trees in $H^{m-1}$ and $g(T_R^i) = (1/2)g(T_k) \geq 1$, we must have \[
        \sum_{j=1}^{k-1} g(T_j) + g(T_R^{m}) \geq (k-1)\cdot 1 + 1 = k,
    \]
    which implies \[
        g(H^m) = \sum_{j=1}^{k-1} g(T_j) + 2g(T_R^m) \geq k+g(T_R^m).
    \]
    This is a contradiction to \Cref{eq:contradictme}.

    Regarding the running time, it is clear that the preprocessing runs in time $O(|E| \log |E|)$ with any implementation of a heap $H$ that can do push and pop operations in $O(\log |H|)$ time. Furthermore, $v$ does not take the same value twice on line $13$, and thus the time used to compute the set to cut each day is linear in $|E|$.
\end{proof}

The following is an immediate corollary of \Cref{lem:kuniform}.
\begin{corollary}
\label{cor:bgt2}
\Cref{alg:fun} implements a height-$2$ schedule for any BGT instance.
\end{corollary}

Furthermore, since any partition matroid is the direct sum of uniform matroids, \Cref{lem:kuniform} also implies an efficient algorithm for partition matroids by running \Cref{alg:fun} on every uniform matroid in parallel.

\begin{corollary}
    \label{cor:partition}
    There is a polynomial-time algorithm implementing a height-2 schedule for any CBGT instance $(E,\cI,g)$, where $(E,\cI)$ is a partition matroid.
\end{corollary}

Observe that if we simply scale the heights of a BGT instance $B$ by some constant $c$, then we also scale the height of any schedule $\pi$ for $B$ by $c$. Furthermore, \Cref{alg:fun} is oblivious to the absolute growth rates of bamboos and only depends on their ratios. This gives the following corollary, which is useful in the next subsections.

\begin{corollary}\label{cor:bgt4}
    For any BGT instance $(E,g)$ where $\sum_{e\in E}g(e)\leq 2$,
    \Cref{alg:fun} implements a schedule $\pi$ with $h(\pi) < 4$.
\end{corollary}

\subsection{Height-4 Schedules for Graphic and Laminar Matroids in Polynomial-Time}
\label{subsec:Graphic_laminar}

In this section, we give polynomial-time algorithms for CBGT on graphic and laminar matroids.
The idea behind both algorithms is to partition the ground set (by means of a coloring) such that each part can be scheduled independently. Define a coloring $c : E \rightarrow [k]$ as \textit{rainbow-circuit-free}\footnote{Also known as rainbow-cycle-forbidding  on graphs~\cite{fujita2010rainbow,Gouge10,hoffman2019rainbow,jendrol200rainbow}.} if $c$ is such that any subset $X \subseteq E$ which contains at most one element of each color is independent. An illustration is shown in \Cref{fig:FUNgraphic}.
It was shown in \cite{berczi2021ListColoring} that one can obtain such colorings where each color class is not too large for several classes of matroids, including graphic and laminar matroids. We are interested in rainbow-circuit-free coloring where, in addition, the total sum of growth-rates within each class is small, and we show that such a coloring can be obtained using the same approach as in \cite[Theorem 1.5]{berczi2021ListColoring}.

Clearly, if we pick at most one element of each color to schedule on every day, we obtain a valid schedule. 
We show that, for any graphic or laminar matroid, one can obtain a rainbow-circuit-free coloring $c$ such that, for every color class $i \in [k]$, we have $\sum_{e \in c^{-1}(i)} g(e) \leq 2$ where $c^{-1}(i) = \{e \in E \mid c(e) = i\}$. 
By \Cref{cor:bgt4}, running a Fuse--Unfuse schedule for each color class in parallel gives a schedule with height less than $4$. We note that the schedule for each part can be any schedule for which \Cref{cor:bgt4} holds. 

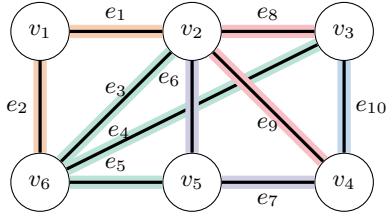
\begin{figure}
\centering
\begin{tikzpicture}[
  every node/.style={font=\footnotesize},
  tnode/.style={circle,draw,minimum size=2em,inner sep=0pt},
  cnode/.style={circle,draw,minimum size=2em,inner sep=0pt},
  edge/.style={draw, line width=1.2pt}
]
\node[tnode] (v1) at (0,2) {$v_1$};
\node[tnode] (v2) at (2,2) {$ v_2 $};
\node[tnode] (v3) at (4,2) {$ v_3 $};
\node[tnode] (v4) at (4,0) {$ v_4 $};
\node[tnode] (v5) at (2,0) {$ v_5 $};
\node[tnode] (v6) at (0,0) {$ v_6 $};

\draw [MyOrange!30, line width = 5pt] (v1) -- (v2) ;
\draw[edge] (v1) to node[midway, above] {$e_1$} (v2);
\draw [MyOrange!30, line width = 5pt] (v1) -- (v6) ;
\draw[edge] (v1) to node[midway, left] {$e_2$} (v6);
\draw [greenish!30, line width = 5pt] (v6) -- (v2) ;
\draw[edge] (v6) to node[midway, above] {$e_3$} (v2);
\draw [greenish!30, line width = 5pt] (v6) -- (v3) ;
\draw[edge] (v6) to node[pos=0.3, left] {$e_4$} (v3);
\draw [greenish!30, line width = 5pt] (v6) -- (v5) ;
\draw[edge] (v6) to node[midway, above] {$e_5$} (v5);
\draw [MyPurple!35, line width = 5pt] (v5) -- (v2) ;
\draw[edge] (v5) to node[pos=0.8, left] {$e_6$} (v2);
\draw [MyPurple!35, line width = 5pt] (v5) -- (v4) ;
\draw[edge] (v5) to node[midway, below] {$e_7$} (v4);
\draw [myred!30, line width = 5pt] (v2) -- (v3) ;
\draw[edge] (v2) to node[midway, above] {$e_8$} (v3);
\draw [myred!30, line width = 5pt] (v2) -- (v4) ;
\draw[edge] (v2) to node[midway, below] {$e_9$} (v4);
\draw [myblue!35, line width = 5pt] (v3) -- (v4) ;
\draw[edge] (v3) to node[midway, right] {$e_{10}$} (v4);
\end{tikzpicture}
\caption{Rainbow-circuit-free coloring of a graphic matroid on edge set $\{e_1, e_2, \dots, e_{10}\}$.}
\label{fig:FUNgraphic}
\end{figure}

\subsubsection{Graphic Matroids}
\label{sec:graphic_Matroid}

Let $(E,\mathcal{I},g)$ be a CBGT instance such that $(E,\mathcal{I})$ is a graphic matroid, and let $G$ be the graph associated with $(E,\mathcal{I})$, with edge set $E(G)$ and vertex set $V(G)$. Let $n = |V(G)|$.

Our goal is to find a rainbow-circuit-free coloring $c$ of $E$ such that no color class has sum of growth rates more than $2$. For a graphic matroid, a rainbow-circuit-free coloring is an edge coloring where no cycle is rainbow, where a cycle $C$ in $G$ is rainbow if all edges on $C$ have distinct colors. That is, in a rainbow-circuit-free coloring, no set of edges with distinct colors forms a cycle. We use the term \emph{circuit} to remain consistent with the matroid terminology.

Colorings that forbid certain rainbow subgraphs have been studied in the literature \cite{fujita2010rainbow,Gouge10,hoffman2019rainbow,jendrol200rainbow}. It is also known that $n-1$ colors suffice to obtain a rainbow-circuit-free coloring. We use the following simple and well-known observation (see, e.g.,~\cite{Gouge10}) and include a proof for completeness.

\begin{restatable}{lemma}{goodcoloring}
    \label{lem:goodcoloring}
    Let $c$ be a coloring of $G$ obtained as follows. Take any partition
    $(A,B)$ of $V(G)$. Color the edges between $A$ and $B$ with color $1$,
    color $G[A]$ with any rainbow-circuit-free coloring using colors
    $\{2,3,\dots,|A|\}$, and color $G[B]$ with any rainbow-circuit-free coloring
    using colors $\{|A|+1,\dots,n-1\}$. Then, $c$ is rainbow-circuit-free.
\end{restatable}

\begin{proof}
    Let $C$ be any cycle in $G$. If $C$ is contained in $G[A]$ or $G[B]$,
    then $C$ is not rainbow since $G[A]$ and $G[B]$ are colored with
    rainbow-circuit-free colorings. Otherwise, $C$ contains at least $2$ edges
    between $A$ and $B$ and is thus not rainbow as all such edges have such
    color $1$.
\end{proof}

An algorithm computing a rainbow-circuit-free coloring in which every color class has sum of growth rates at most $2$ is given by \Cref{alg:graphic}. 
The high-level idea of the algorithm is as follows. 
At every step, pick a vertex with smallest ``vertex growth'' $g_{H}(v)$ (line 4) in the current subgraph $H$ of $G$ induced by the vertices that have not been picked yet and then color all the adjacent edges with a new color. By \Cref{lem:goodcoloring}, the resulting coloring is rainbow-circuit-free. Furthermore, using \eqref{eq:growth-and-rank}, we show that each color class satisfies the growth-rate bound.

\begin{algorithm}
\caption{Rainbow-circuit-free Coloring of a Graphic Matroid}
\label{alg:graphic}
 \hspace*{\algorithmicindent} \textbf{Input:} A CBGT instance $(E,\I,g)$  where $(E,\I)$ is the graphic matroid. \\
 \hspace*{\algorithmicindent} \textbf{Output:} A rainbow-circuit-free coloring $c:E\to[n-1]$.
\begin{algorithmic}[1]
\State $H_{n-1} \gets G$.
\State $i \gets n-1$.
\While{$i \ge 1$}
    \State Compute the vertex growth rates $g_{H_i}=\sum_{u: vu \in E(H_i)} g(vu)$.
    \State Select a vertex $v_i \in V(H_i)$ minimizing $g_{H_i}(v)$.
    \State Assign color $i$ to all edges incident to $v_i$.
    \State $H_{i-1} \gets H_i \setminus \{v_i\}$.
    \State $i \gets i-1$.
\EndWhile
\end{algorithmic}
\end{algorithm}

\begin{restatable}[$\star$]{thm}{THMgraphic}
    \label{thm:graphic}
    Let $(E,\I,g)$ be an CBGT instance where $(E,\I)$ is the graphic matroid. Then there exists a rainbow-circuit-free coloring $c:E\to[n-1]$ of $G$
    such that $\sum_{e \in E_i} g(e) \leq 2-\frac{2}{n}$ for every color
    class $E_i$ of $c$. Furthermore, $c$ can be obtained in polynomial time.
\end{restatable}

\Cref{thm:graphic} and \Cref{cor:bgt4} immediately imply a polynomial-time algorithm implementing a schedule of height less than $4$ for graphic matroids.
See \Cref{example:FUN-graphic}.
\begin{restatable}[$\star$]{thm}{THMgraphicFUN}
\label{thm:graphic-FUN}
    For any CBGT instance $(E,\I,g)$, where $(E,\I)$ is a graphic matroid, there is a schedule $\pi$ with $h(\pi)< 4$. Furthermore, $\pi$ can be implemented in polynomial time.
\end{restatable}

\begin{example}
    \label{example:FUN-graphic}
    Consider the instance of a graphic matroid $M$ on edge set $\{e_1, e_2, \dots, e_{10 }\}$ with $g(e)=1/2$ for every edge. A rainbow-circuit-free coloring obtained by \Cref{alg:graphic} is shown in \Cref{fig:FUNgraphic}, from which we obtain the $n-1=5$ schedules $\pi_1=(e_1,e_2,\dots)$, $\pi_2=(e_3,e_4,e_3,e_5,\dots)$, $\pi_3=(e_6,e_7, \dots)$, $\pi_4=(e_8,e_9,\dots)$, and $\pi_5=(e_{10}, \dots)$. The schedule for $M$ given by \Cref{thm:graphic-FUN} is $\pi=(\{e_1,e_3,e_6,e_8,e_{10}\},\{e_2,e_4,e_7,e_9,e_{10}\},\{e_1,e_3,e_6,e_8,e_{10}\}, \{e_2,e_5,e_7,e_9,e_{10}\} ,\dots)$. 
\end{example}

\subsubsection{Laminar Matroids}
\label{sec:laminar_Matroid}

Through this section, let $M=(E,\mathcal{I})$ be a laminar matroid defined by the laminar family $\mathcal{L}$ and the function $b:\mathcal{L} \rightarrow \N$. We start with a useful characterization of the rainbow-circuit-free colorings of $M$.

\begin{restatable}{lemma}{rainbowlaminar}
    \label{lem:rainbowlaminar}
    Let $M$ be a laminar matroid. Let $c:E\to \mathbb{N}$ be any coloring
     such that $|c(L)| \leq b(L)$ holds for all $L \in \mathcal{L}$. Then
    $c$ is a rainbow-circuit-free coloring of $M$.
\end{restatable}
\begin{proof}
    Let $C$ be any circuit in $M$. Then there exists $L \in
    \mathcal{L}$ such that $|C \cap L| > b(L)$. Since at most $b(L)$ distinct
    colors appear on the elements of $L$, $C$ contains two
    elements with the same color. 
\end{proof}

As in the last subsection, we show that we can construct a rainbow-circuit-free coloring of the ground set such that each color class has growth rate at most $2$. For this, we need the following definition. To \emph{set-fuse} of two elements $e,f \in E$ is to replace $e$ and $g$ by a new bamboo $e \hatcirc f$ with $g(e \hatcirc f) = g(e)+g(f)$. 

Our algorithm is given by \Cref{alg:laminar}. 
The high-level idea is as follows. 
If every set $L \in \mathcal{L}$ satisfies $|L| \leq b(L)$, then we may assign distinct colors to all elements $e \in E$. 
By \Cref{lem:rainbowlaminar}, this coloring is rainbow-circuit-free.    
Otherwise, the algorithm repeatedly set-fuses elements so as to reduce the size of sets in $\mathcal{L}$ until every set $L \in \mathcal{L}$ has size at most $b(L)$. 
After performing all set-fusion operations, we assign distinct colors to the resulting elements. 
To avoid ambiguity, we denote by $S$ the set of elements obtained through set-fusion.

    \begin{algorithm}
    \caption{Rainbow-circuit-free Coloring of a Laminar Matroid}
    \label{alg:laminar}
     \hspace*{\algorithmicindent} \textbf{Input:} A CBGT instance $(E,\I,g)$  where $(E,\I)$ is the laminar matroid. \\
     \hspace*{\algorithmicindent} \textbf{Output:} A rainbow-circuit-free coloring $c:E\to \mathbb{N}$.
    \begin{algorithmic}[1]
        \State $S\gets E$.
        \While{$\mathcal{L} \neq \emptyset$}
            \State Select a set $L \in \mathcal{L}$ that is inclusion-wise minimal.
            \If{$|L| \leq b(L)$}
                \State Remove $L$ from $\mathcal{L}$.
            \Else
            \State Choose distinct elements $s_1,s_2 \in L$ minimizing $g(s_1)+g(s_2)$.
            \State Replace $s_1$ and $s_2$ by $s_1 \hatcirc s_2$ in every set $X \in \mathcal{L} \cup \{S\}$ such that $\{s_1,s_2\} \subseteq X$.
            \EndIf
        \EndWhile
        \For{$s \in S$}
            \State Assign a new color to $s$.
            \If{$s = s_1 \hatcirc s_2$}
                \State Iteratively set-unfuse $s$ and assign the same color to all elements of $E$ obtained.
            \EndIf
        \EndFor
    \end{algorithmic}
    \end{algorithm}

\begin{restatable}[$\star$]{thm}{rainbowfreelaminar}
    \label{thm:rainbowfreelaminar} 
    Let $(E,\I,g)$ be an CBGT instance where $(E,\I)$ is a laminar matroid. 
    Then, there exists a rainbow-circuit-free coloring $c:E\to \mathbb{N}$ of $M$ such that $\sum_{e \in E_i} g(e) \leq 2$ holds for every color class $E_i$ of $c$. Furthermore, $c$ can be obtained in polynomial time.
\end{restatable}

\Cref{thm:rainbowfreelaminar} immediately implies a simple schedule of height at most $4$. 

\begin{restatable}[$\star$]{thm}{thmLaminarFUN}
\label{thm:laminar-FUN}
    For any CBGT instance $(E,\I,g)$, where $(E,\I)$ is the laminar matroid, there is a schedule $\pi$ with $h(\pi)< 4$. Furthermore, $\pi$ can be implemented in polynomial time.
\end{restatable}

\section{General Set Systems}
\label{sec:general}

In this section, we give an efficient implementation of a height-$O(\log n)$ schedule for any set system, even those that are not independence systems. This improves upon a result from \cite{cicerone2019FairHitting}, which gives a randomized algorithm achieving the same guarantee (but only with high probability) and with running time polynomially dependent on $|\I|$\footnote{Recall that an $O(\log n)$-approximation algorithm, which is the focus of \cite{cicerone2019FairHitting}, can be obtained from our algorithm by scaling $g$ appropriately as described in~\Cref{sec:approxfact}.}. Our algorithm assumes that $\I$ is given implicitly via access to a maximum-weight independent-set oracle, and makes at most $2|E|$ calls to this oracle. We remind the reader that a maximum-weight independent-set oracle takes weights $w_e \geq 0$ for $e \in E$ and returns a set $I \in \I$ with maximum $\sum_{e \in I} w_e$.

Throughout this section, we use $n=|E|$. 

\subsection{Upper Bound}
\label{ss:log-upper-bound}
In this subsection, we show that a height-$O(\log n)$ schedule exists for any instance. We give a proof here (even though it was already shown in \cite{cicerone2019FairHitting}) as warm-up for our efficient algorithm. 
We begin by introducing some notation.
Let $c> 2$. 
We partition the elements into \emph{slow} and \emph{fast} sets according to the threshold $\tau_n = c\frac{\ln n}{n}$. Formally:
\begin{align*}
    E_{\mathrm{slow}} & = \{ e\in E \mid g(e) \leq \tau_n\}, \\
    E_{\mathrm{fast}} & = \{ e \in E \mid g(e) > \tau_n \}.
\end{align*}

The main idea is to devise two separate schedules, serving fast and slow elements separately.
For the fast elements, we use the probabilistic method to show that a schedule generated by the coefficients of the convex combination suffices.

\begin{lemma}
    \label{lem:fast-schedule}
    There exists an infinite valid schedule $\pi_{\mathrm{fast}}$ such that the height of every fast element is bounded by $2c \ln n$ for all $t \ge 1$. 
    That is,
    \[ \max\limits_{t \ge 1, e \in E_{\mathrm{fast}}}h^{\pi_{\mathrm{fast}}}_t(e) \leq 2c \ln n .\]
\end{lemma}
\begin{proof}
    Let $\lambda \colon \cI \to [0,1]$ be such that $\sum_{I\in\cI} \lambda(I)=1$ and $g(e) \leq \sum_{I\in\cI \colon e \in I} \lambda(I).$
    
    First, we prove the existence of a finite schedule $\sigma$ of length $n$ where the maximum height is bounded by $    
    H = c \ln n$.
    We select a sequence of $n$ independent sets $\sigma = (I_1, \dots, I_n)$ randomly, where each $I_t$ is chosen independently from $\cI$ according to the probability distribution $\lambda$.

    Consider a single element $e \in E$ and a time step $t \in [n]$.
    The height $h^\sigma_t(e)$ exceeds $H$ only if $e$ was not cut in the preceding $H/g(e)$ steps. 
    Thus,
    \begin{align*}
        \PP_\sigma[h^\sigma_t(e) > H] & \leq (1 - g(e))^{H/g(e)} \\
        & \leq e^{-H} = e^{-c \ln n} = n^{-c}.
    \end{align*}
    By applying a union bound over all elements $e \in E$ and all time steps $t \in [n]$, the probability that any element exceeds height $H$ at any timestep is
    \[ \PP_\sigma\left[\max_{e \in E, t\in [n]} h^S_t(e) > H\right] \leq \sum_{e \in E} \sum_{t=1}^n \PP_\sigma[h^S_t(e) > H] \leq n^2 \cdot n^{-c} = n^{2-c}. \]
    Since $c > 2$, this probability is strictly less than 1. 
    Therefore, there exists a sequence $\sigma = (I_1, \dots, I_n)$ such that the height of any element within the schedule is at most $H$.

    We construct the infinite schedule $\pi_{\mathrm{fast}}$ by concatenating infinite copies of $\sigma$. 
    Consider a fast element $e \in E_{\mathrm{fast}}$. 
    By definition, $g(e) > \frac{c \ln n}{n} = H/n$, which implies $H/g(e) < n$. 
    Since the height of $e$ inside $\sigma$ never exceeds $H$, the maximum gap between consecutive cuts of $e$ in $\sigma$ is at most $H/g(e) < n$. 
    This guarantees that $e$ is cut at least once in every copy of $\sigma$.

    In the concatenated schedule $\pi_{\mathrm{fast}}$, the maximum gap between two cuts is bounded by the gap across the boundary of two copies of $\sigma$, which is at most $2 \cdot \frac{H}{g(e)}$. 
    Consequently, the maximum height of a fast element is bounded by
    \[ \max_{t\in \NN, e\in E_{\mathrm{fast}}} h^{\pi_{\mathrm{fast}}}_t(e) \le g(e) \cdot \left( \frac{2H}{g(e)} \right) = 2H = 2c \ln n. \qedhere\]
\end{proof}

For the slow elements, we serve them using a Round-Robin strategy. 
Since their growth rate is low, we only need to service them every $O(n)$ timesteps to keep the height below $O(\log n)$. 
We obtain the main result by interleaving the two schedules.

\begin{thm}\label{thm:exub}
    For every $c>2$, there exists a schedule $\pi$ such that $h(\pi)\leq 3c \ln n$.
\end{thm}
\begin{proof}
    We employ a weighted interleaving strategy. 
    Let $\pi_{\mathrm{fast}}$ be the infinite schedule from \Cref{lem:fast-schedule}.
    Let $\pi_{\mathrm{slow}}$ be a Round-Robin schedule of the $n$ elements (i.e., it cuts $I_1,\dots, I_n$ cyclically such that $e_1 \in I_1,\dots , e_n \in I_n$).
    We construct a combined schedule $\pi$ by using the fast schedule $2/3$ of the time and the slow schedule the remaining third.
    Formally, for time steps $t \in \NN$:
    \begin{itemize}
        \item If $t \not= 0 \pmod 3$, we execute the next move from $\pi_{\mathrm{fast}}$.
        \item If $t = 0 \pmod 3$, we execute the next move from $\pi_{\mathrm{slow}}$.
    \end{itemize}

    In the schedule $\pi_{\mathrm{fast}}$, the height is bounded by $2c\ln n$.
    Thus, the height is bounded by:
    \[  \frac{3}{2} H=  3c \ln n. \]
    In $\pi$, the schedule $\pi_{\mathrm{slow}}$ is executed every third step. 
    Since $\pi_{\mathrm{slow}}$ has a period of $n$, every element is guaranteed to be cut at least once every $3n$ steps. 
    For any $e \in E_{\mathrm{slow}}$, the growth rate is $g(e) \le \tau_n$. 
    Thus, the maximum height for any $e \in E_{\mathrm{slow}}$ is at most
    \[ 3n \cdot \tau_n = 3n \cdot \frac{c \ln n}{n} = 3c \ln n. \]
    We conclude that the schedule $\pi$ has a height of at most $3c\ln n$.
\end{proof}

\subsection{Lower Bound}
\label{ss:log-lower-bound}

We now show that the previous schedule is tight (up to constant factors).
The main idea is that every instance where the growth rate is bounded from below and that has a schedule with low height must have a small covering by independent sets. 
Then, we try to construct instances with high growth rates that are hard to cover.

\begin{thm}\label{thm:lower}
    For every $k \in \NN$ there is a CBGT instance with $n=\binom{2k}{k}$ elements for which every infinite valid schedule has height at least $\frac{1}{4}\log_2 n$. 
\end{thm}
\begin{proof}
    Let the set of elements $E$ be the collection of all subsets of $[2k]$ with size $k$; i.e., 
    \[
    E = \{S \subseteq [2k] \mid |S| = k \}.
    \]
    The collection of feasible cuts $\mathcal{I}$ consists of sets containing a fixed element from the universe $[2k]$. 
    Specifically, for each $i \in [2k]$, we define:
    \[
    I_i = \{S \in E \mid i \in S \}.
    \]
    We let $\mathcal{I}$ be the downward closure of $\{I_1, I_2, \dots, I_{2k}\}$.

    We consider the growth vector given by the coefficients $\lambda(I_i) = \frac{1}{2k}$ for all $i \in [2k]$. 
    These coefficients sum to $\sum_{i=1}^{2k} \lambda(I_i) = 1$, and for any element $S \in E$:
    \[
    g(S) = \sum_{i \in [2k] \colon S \in I_i} \lambda(I_i) = \sum_{i \in S} \frac{1}{2k} = \frac{|S|}{2k} = \frac{k}{2k} = \frac{1}{2}.
    \]

    We now show that for any infinite valid schedule, there exists an element $S^* \in E$ that is not cut during the first $k$ time steps.
    Let the schedule select the cuts $I_{i_1}, I_{i_2}, \dots, I_{i_k}$ in the first $k$ steps. 
    Let $J = \{i_1, \dots, i_k\} \subseteq [2k]$ and note that $|J| \leq k$.
    Therefore, there exists a subset $S^* \subseteq [2k] \setminus J$ with $|S^*| = k$.
    Furthermore, this set has not been cut as $J\cap S^* =\emptyset$.
    
    The height of $S^*$ at time $k$ is:
    \[
    h_k(S^*) = k \cdot g(S^*) = \frac{k}{2}.
    \]
    Using the bound $n = \binom{2k}{k} < 2^{2k}$, we have that $k>\frac{1}{2}\log_2 n$.
    Substituting this into the previous equation we obtain
    \[
    h_k(S^*) > \frac{1}{2} \left( \frac{1}{2} \log_2 n \right) = \frac{1}{4} \log_2 n. \qedhere
    \]
\end{proof}

A similar, albeit more involved, construction via algebraic methods can push the constant in the lower bound.

\jamisonlb

\subsection{Efficient Implementation}
\label{ss:logboundimplement}

We now show how to efficiently implement the schedule from~\Cref{ss:log-upper-bound}. Note that we can find some $I \in \I$ with $e \in I$ by calling the maximum-weight independent-set oracle with $w_e = 1$ and $w_{e'} = 0$ for all $e' \neq e$. (If $\I$ is downward-closed, we can just take $I = \{e\}$.) Thus, computing $\pi_{\mathrm{slow}}$ requires only one oracle call per $e \in E_{\mathrm{slow}}$.

\begin{lemma}
    \label{lem:slow-efficient}
    For any CBGT instance $(E,\cI,g)$, the schedule $\pi_{\mathrm{slow}}$ can be computed by making $n$ queries to a maximum-weight independent-set oracle for $\I$.
\end{lemma}

We next describe how to implement an appropriate schedule for the fast elements.
Given a height vector $h \in \QQ_{\geq 0}^{E}$, we define the \emph{potential} of $h$ as $\Phi(h) = \sum_{e \in E_{\mathrm{fast}}} \exp (h(e))$.
Furthermore, if we are given a height vector $h\in \QQ_{\geq 0}^{E}$ and an independent set $I\in \cI$, we define \emph{the potential of $h$ after cutting $I$} as
\[
\Phi(h, I) = \sum_{e \in E_{\mathrm{fast}}\setminus I} \exp(h(e)+g(e)) +  \sum_{e \in E_{\mathrm{fast}}\cap I} \exp(g(e)). 
\]
Given a height vector $h\in \QQ_{\geq 0}^E$, one can find an $I \in \I$ minimizing $\Phi(h, I)$ using the maximum-weight independent-set oracle. Indeed, observe that for any $I' \subseteq E$ and any $e \notin I'$ we have \[
\Phi(h, I' \cup \{e\}) = \begin{cases}
    \Phi(h, I')+\exp(g(e))-\exp(g(e)+h(e)) & \text{if $e \in E_{\mathrm{fast}}$}, \\
    \Phi(h, I') & \text{otherwise}.
\end{cases} 
\]
Thus, the problem of minimizing $\Phi(h, I)$ over $I \in \I$ is a maximum-weight independent-set problem when we set $w_e=\exp(g(e)+h(e))-\exp(g(e)) \geq 0$ for $e \in E_{\mathrm{fast}}$ and $w_e = 0$ for $e \in E \setminus E_{\mathrm{fast}}$.

We show that the algorithm that greedily chooses a cut minimizing the potential at every time step (see \Cref{alg:potential}) is an appropriate schedule for the fast elements.

\begin{algorithm}
\caption{Greedy Potential Scheduling}
\label{alg:potential}
 \hspace*{\algorithmicindent} \textbf{Input:} A CBGT instance $(E,\cI,g)$. \\
 \hspace*{\algorithmicindent} \textbf{Output:} $\pi = (\pi_1,\pi_2,\dots, \pi_n)$ a valid schedule of length $n$.
\begin{algorithmic}[1]
    \State $h\gets (0,\dots,0)$
    \For{$t=1,\dots,n$}
    \State Choose $I\in\cI$ minimizing $\Phi(h,I)$.
    \State $\pi_t \gets I$
    \For{$e\in E$}
        \State If $e\notin I$, then $h(e) \gets h(e)+g(e)$. 
        Otherwise, $h(e)\gets g(e)$. 
    \EndFor
    \EndFor
\end{algorithmic}
\end{algorithm}

\begin{lemma}
    \label{lem:fast-potential}
    For $n\geq 3$, the schedule computed by \Cref{alg:potential} keeps the height of fast elements below $4\ln n$ during the first $n$ timesteps; i.e., 
    \[
    \max_{e \in E_{\mathrm{fast}}, t \in [n]} h_t(e) \leq 4 \ln n. 
    \]
\end{lemma}
\begin{proof}
    Let $\mathbf{0}=h_0,h_1,\dots,h_n \in \QQ_{\geq 0}^E$ be the $n+1$ height vectors obtained during the execution of \Cref{alg:potential}.
    We define $\rho_n = 1 - \tau_n^3 < 1$.
    Let $\lambda \colon \cI \to [0,1]$ be such that $\sum_{I\in\cI} \lambda(I)=1$ and $g(e) \leq \sum_{I\in\cI \colon e \in I} \lambda(I).$
    
    Note that $\Phi(h_{t+1}) = \min_{I\in \cI} \Phi(h_t,I) \leq \EE_{I}(\Phi(h_t,I))$, where we consider the random choice of an independent set $I\in \cI$ according to the probability distribution $\lambda$.
    Thus,
    \begin{align*}
        \Phi(h_{t+1}) & \leq \EE_I(\Phi(h_t,I)) \\ 
        & = \sum_{e \in E_{\mathrm{fast}}} \exp(g(e))\PP_I[e \in I] + \sum_{e \in E_{\mathrm{fast}}} \exp(h_t(e)+g(e)) \PP_I[e \notin I] \\ 
        & \leq \sum_{e \in E_{\mathrm{fast}}} e + (1-g(e))\exp(g(e)) \exp(h_t(e)) . \numberthis \label{eq:potential-bound}
     \end{align*}
    Since $\exp(z) \leq 1+z+z^2$ for $z \in [0,1]$, we obtain
    \begin{equation}
        \label{eq:growth-bound}
        (1-g(e))\exp(g(e)) \leq (1-g(e))(1+g(e)+g(e)^2) = 1-g(e)^3 \leq \rho_n.
    \end{equation}
    Putting \eqref{eq:potential-bound} and \eqref{eq:growth-bound} together, we obtain
    \[
    \Phi(h_{t+1}) \leq en + \rho_n \Phi(h_{t}).
    \]
    Applying the previous inequality repeatedly, for every $t\in [n]$ we obtain that
    \[
    \Phi(h_{t}) \leq \Phi(h_0)\rho_n^n + en \sum_{i=1}^{n-1} \rho_n^i  \leq n + \frac{en}{1-\rho_n} \leq \frac{1-\rho_n+e}{1-\rho_n} n \leq \frac{(1+e)n}{\tau_n^3}.
    \]
    Using that $1+e \leq (c\ln n)^3$ for $n\geq 3$, we conclude that for every $e\in E_{\mathrm{fast}}$ and $t \in [n]$, it holds that
    \[
    h_t(e) \leq \ln (\Phi(h_t)) \leq \ln\bigg(\frac{(1+e)n^4}{(c\ln n)^3}\bigg) \leq 4 \ln n.
    \]
    This completes the proof.
\end{proof}

Interleaving the schedule of \Cref{lem:slow-efficient} with the one of \Cref{lem:fast-potential} we obtain the following.

\compub

\begin{proof}
    If $n\geq 3$ the interleaving of the schedules in \Cref{lem:slow-efficient,lem:fast-potential} gives the desired result. 
    If $n < 3$, the schedule $\pi_{\mathrm{slow}}$ of \Cref{lem:slow-efficient} has a guarantee of $2$.
\end{proof}

\section{Open Problems}
\label{sec:open}

We conclude with some open problems.

Let us call a CBGT instance $(E,\I,g)$ a matroid instance if $(E,\I)$ is a matroid.
We have shown that any matroid CBGT instance has a schedule $\pi$ with $h(\pi) < 2$. However, the implementation suggested by our proof uses time exponential in the number of bits required to represent $g$.
Ideally, we would like an algorithm which is given access to $\I$ through some reasonable oracle (e.g. an indepence oracle) and which implements a schedule $\pi$ with $h(\pi) < 2$ (or even $d(\pi) < 1$) in time polynomial in both $|E|$ and the number of bits required to represent $g$. 

\begin{restatable}{problem}{algorithm}
    Does there exists a polynomial-time implementation of a height-$2$ schedule for matroid CBGT instances when access to $\I$ is given through some natural oracle?
\end{restatable}

The notion of a \emph{$\ell$-system} was introduced by Jenkyns in the 1970s \cite{Jenkyns76}. An independence systems $(E,\I)$ is a $\ell$-system if, for any $S \subseteq E$, the ratio of the smallest to the largest maximal independent subset of $S$ is at most $\ell$. 
Examples of $\ell$-systems include the set systems obtained from the intersection of $\ell$ matroids \cite{Korte78}. 
Moreover, matroids are exactly the $1$-systems. 
As we have shown that the family of matroid CBGT instances admit height-$2$ schedules, it is natural to ask the following.

\begin{restatable}{problem}{ksystem}
    Is there a function $f \colon \NN\to\NN$ such that all $\ell$-system CBGT instances admit an $f(\ell)$-height schedule?
\end{restatable}

We note that the lower bound presented in \Cref{thm:lower} is a $\Theta(|E|)$-system.
Indeed, for the set system $(E,\cI)$ of \Cref{thm:lower}, consider the subset 
\[
S = \{ X \in E \mid 1 \in X \text{ and } 2\notin X \}\subseteq E
\]
and choose $X_1,X_2\in E$ such that $X_1\cap X_2 = \{2\}$. 
Note that $\{X_1,X_2\}$ and $\{X \in S \mid 1 \in X \}$ are the smallest and largest maximal independent sets respectively.  
Since their sizes are $2$ and $\binom{2k-2}{k-1} \geq \frac{1}{4} \binom{2k}{k} = \frac{1}{4}|E|$ respectively, we conclude that $(E,\cI)$ is a $\Theta(|E|)$-system.
This implies that if such an $f$ exists, we must have $f(\ell) = \Omega(\log \ell)$.

\section*{Acknowledgements}

We are grateful to Tobias Mömke, Victor Verdugo, and José Verschae for organizing the \emph{Santiago Summer Workshop on Combinatorial Optimization 2025}, where this research began. 
We would like to give special thanks to Sofía Errázuriz for valuable discussions at the beginning of this project.
Finally, we thank all participants of the workshop for inspiring discussions.

Mirabel Mendoza-Cadena was supported by Centro de Modelamiento Matemático (CMM) BASAL fund FB210005 for center of excellence from ANID-Chile.
Arturo Merino was supported by ANID FONDECYT Iniciación No. 11251528.
Kevin Schewior was supported by the Independent Research Fund Denmark, Natural Sciences, grant DFF-4283-00079B.

\bibliographystyle{plain}
\bibliography{bib}

\newpage
\appendix
\label{sec:appendix}

\section{Proofs for \texorpdfstring{\Cref{sec:prelim}}{Section~\ref{sec:prelim}}}

For the proof of \Cref{thm:optlb}, we use the following simple observation, which was originally made in the context of PS and since also used for BGT (see \cite[Thm 2.1]{HMRTV89} for the PS version).

\begin{restatable}{obs}{obsfinite}
\label{obsfinite}
    For any infinite schedule $\pi$ with finite $h(\pi)$, there exists a finite subsequence $\pi'$ of $\pi$ such that $h(\pi'^\infty) \leq h(\pi)$.
\end{restatable}

\begin{proof}
    Call the vector $(h_t^\pi(e))_{e \in E}$ the \textit{state} at time $t$ for every $t \in \N$. Observe that, as $h(\pi)$ is finite and $h^\pi_t(e)$ is always a multiple of $g(e)$, there are finitely many possible states. Thus, some state must repeat. Let $t < t'$ be time steps such that $(h_t^\pi(e))_{e \in E} = (h_{t'}^\pi(e))_{e \in E}$. The schedule $\pi(1),\dots,\pi(t-1),(\pi(t),\dots,\pi(t'-1))^\infty$ has the same height as $\pi$ as it is identical up to time $t'-1$ and then cycles through a subset of already visited states. The final observation is that the schedule $\pi' = (\pi(t),\dots,\pi(t'-1))^\infty$ cannot be worse than $\pi$ as we would have $h(\pi') = h(\pi)$ if we started $\pi'$ from the state $(h_t^\pi(e))_{e \in E}$, and it is easy to show by induction on time that starting $\pi'$ from the state $(0)_{e \in E}$ cannot increase the height.
\end{proof}
\optlb*
\begin{proof}
    Suppose that $\Lambda(g) = 1$ and yet $h(\pi^\star) < 1$. We may assume that $\pi^\star = \pi^\infty$ for some finite schedule $\pi$ by \Cref{obsfinite}. For $I \in \I$, let \[
    \lambda(I) = \frac{|\{t \in [|\pi|] \mid \pi(t) = I\}|}{|\pi|}
    \] be the fraction of time steps in which $I$ is cut. 
    Then $\sum_{I \in \I} \lambda(I) = 1$. Since $h^\pi(e) = \gamma^\pi(e)g(e) < 1$, we must have $\gamma^\pi(e) < 1/g(e)$, implying that $e$ is cut in more than a $g(e)$-fraction of the time steps. That is, $\sum_{I \in \I:e \in I} \lambda(I) > g(e)$. But now $\lambda$ can be scaled by $\min_{e \in E} \frac{g(e)}{\sum_{I \in \I:e \in I} \lambda(I)} < 1$ while still satisfying the constraints of the LP defining $\Lambda(g)$ (\Cref{fig:LPprimal-left}), contradicting that $\Lambda(g) = 1$.
\end{proof}

\cordensity*

\begin{proof}
     As $P$ has density at most $1/2$, there exists a $\lambda \colon \I \rightarrow [0,1]$ such that $\sum_{I \in \I} \lambda(I) \leq 1/2$ and $\sum_{I \in \I: e \in I} \lambda(I) \geq 1/a(e)$ for all $e \in E$. 

    Let $B = (E,\I,g)$ be the CBGT instance where $g(e) = 2/a(e)$. 
    This way, $a(e) = \lfloor 2/g(e) \rfloor$. 
    Thus, by \Cref{prop:pinwheel}, it suffices to show that $B$ has an infinite valid schedule $\pi$ with $h(\pi) \leq 2$. Let $\lambda'(I) = 2\lambda(I)$ for all $I \in \I$. 
    Then $\sum_{I \in \I} \lambda'(I) = 2 \sum_{I \in \I} \lambda(I) \leq 1$, and so $g \in \conv(\{\ind_I \mid I \in \I\})$. 
    Now, by \Cref{thm:matroidtwo}, there exists a schedule $\pi$ for $B$ with $h(\pi) < 2$, which completes the proof.
\end{proof}

\corddensityup*

\begin{proof}
    Let $n$ be any integer. Pick $n' \leq n$ so that $n' = \binom{2k}{k}$ for some $k \in \N$. Note that $n' \geq n/4$.
    By \Cref{thm:lower}, there exists an instance $(E,\I,g)$ of CBGT with $|E| = n'$ which does not have any infinite valid schedule of height $(1/4)\log_2 n'$. Thus, by \Cref{prop:pinwheel} the instance $(E,\I,a)$ where $a(e) = \lfloor c/g(e)\rfloor$ for $c = (1/4)\log_2 n'$ for all $e \in E$ is not schedulable. 
    We now show that the density of $(E,\I,a)$ is $O(1/\log n)$, which will complete the proof. 

    For this, let $\lambda \colon E \rightarrow [0,1]$ be such that $\sum_{I \in \I} \lambda(I) = 1$ and $g(e) = \sum_{I \in \I : e \in I} \lambda(I)$. Let $\lambda'(I) = 2\lambda(I)/c$. 
    Then, for any $e\in E$ it holds that
    $$\sum_{I\in\I:e \in I} \lambda'(I) = \frac{2}{c}\cdot \sum_{I\in\I:e \in I} \lambda(I)=\frac{2}{c}\cdot g(e)\geq \frac{1}{a(e)},$$ where we used $a(e)=\lfloor c/g(e)\rfloor\geq c/(2 g(e))$ for the inequality. Also, we have that 
    \[
    \sum_{I\in\I}\lambda'(I)=\sum_{I\in\I}\frac2c\cdot\lambda(I)=\frac2c\cdot \sum_{I\in\I}\lambda(I)=\frac2c\in O\left(\frac{1}{\log n}\right),\] 
    which implies that the density of $(E,\I,a)$ is $O(1/\log n)$.
\end{proof}

\section{Proofs for \texorpdfstring{\Cref{sec:FUN-algorithms}}{Section~\ref{sec:FUN-algorithms}}}
\label{app:FUN}

\THMgraphic*
\begin{proof}
    We first introduce some notation.
    For any subgraph $H$ of $G$, we define a growth rate on the vertices of $H$ as  \[
        g_{H}(v) = \sum_{u: vu \in E(H)} g(vu).
    \]
    We claim that \Cref{alg:graphic} computes the desired coloring. 
  
    For a graphic matroid $r(X)$ is the size (number of edges) in a maximal forest in $G[X]$, which is at most $|V(G[X])|-1$, and hence \eqref{eq:growth-and-rank} is equivalent to 
     \[
        \sum_{e \in X} g(e) \leq |V(G[X])|-1\ \ \ \ \forall X \subseteq E.
    \]
    
    Let $c$ be the coloring obtained by \Cref{alg:graphic}.
    
    We claim that $c$ is a rainbow-cycle-free coloring.
    First, it is clear that the algorithm runs $n-1$ times and we assign exactly one color to each edge. It follows that $c$ is indeed a coloring.
    We show that $c$ is rainbow-circuit-free inductively. 
    First, consider the partition $(V(H_{n-2}), V(H_{n-1})\setminus V(H_{n-2}))$ of $G=H_{n-1}$. 
    By \Cref{lem:goodcoloring}, we need to show that (1) all the edges between the two sets should be colored by one color, (2) there is a rainbow-circuit-free coloring for $H_{n-1}[V(H_{n-2})]$, and  (3) there is a rainbow-circuit-free coloring for $H_{n-1}[V(H_{n-1})\setminus V(H_{n-2})]$. Property (1) holds because $V(H_{n-1})\setminus V(H_{n-2})=\{v_i\}$, and $c$ colored all its adjacent edges by one color. 
    Property (3) trivially holds because $V(H_{n-1})\setminus V(H_{n-2})$ is a singleton. We need to show (2), but this holds by induction. 
    This concludes the claim. 

    Thus we only need to show that for every $i\in [n-1]$, the vertex $v_i$ chosen in step $i$ satisfies $\sum_{e \in E_i} g(e) =\sum_{u:v_iu\in E(H_i)} g(v_iu)=g_H(v_i) \leq
    2 - \frac{2}{n}$. To see this, recall that $\sum_{e \in E(H)} g(e) \leq |V(H)|-1$ holds for any subgraph $H$ of $G$. Thus, the average of $g_H(u)$
    over all $u \in V(H)$ is at most $2(|V(H)|-1)/|V(H)|$. In particular, for $i\in [n-1]$ 
    \[
         g_{H_i}(v_i) \leq \frac{2(|V(H_i)|-1)}{|V(H_i)|} = 2 - \frac{2}{|V(H_i)|} \leq
        2-\frac{2}{n},
    \]
    as desired.

    For the running time, note that the While-loop runs at most $n-1$ times, and the rest of the operations can be obtained in polynomial time. This completes the proof.
\end{proof}

\THMgraphicFUN*
\begin{proof}
    The schedule $\pi$ is constructed from $n-1$ disjoint Fuse--Unfuse schedules $\pi_i$.
    By \Cref{thm:graphic}, we obtain a rainbow-circuit-free coloring $c$ of $G$ such that each color class $E_i$ satisfies that the sum of growth rates is less than $2$, $i\in[n-1]$. For each $i\in[n-1]$, let $\pi_i$ be the Fuse--Unfuse schedule for the instance $(E_i,g_i)$ where $g_i$ is the growth rate function induced by $g$ on $E_i$. We define $\pi$ as follows: At time $t$, cut the set $I_t=\cup_{i \in[n-1]}\pi_i(t)$. As we pick an element of each color class, and the coloring is rainbow-circuit-free, we obtain that $I_t$ is a forest, and thus the schedule is valid.
    
    For the height, observe that by \Cref{cor:bgt4} we have that $h(\pi_i) < 4$, and hence $h(\pi) = \max_i h(\pi_i) < 4$ as required.

    Finally, the coloring $c$ can be computed in polynomial time by \Cref{thm:graphic}, and each of the $n-1$ instances can be constructed in polynomial time using \Cref{cor:bgt4}. 
    Consequently, the resulting schedule can be implemented in polynomial time, which completes the proof.
\end{proof}

\rainbowfreelaminar*
\begin{proof}
    Observe that for laminar matroids, \eqref{eq:growth-and-rank} implies that for every $L \in \mathcal{L}$ it holds that
    \[\sum_{e \in L}g(e)\leq b(L).\] 
    
    We claim that \Cref{alg:laminar} computes the desired coloring. 
    Recall that we denote by $S$ the set of elements obtained through set-fusion.
    Let $c$ be the coloring obtained by \Cref{alg:laminar}.
    
    We begin by observing that $c$ is a well-defined coloring. Indeed, if $e \in E \cap S$, then $e$ is assigned exactly one color. If $e$ is set-fused during the execution of the algorithm, let $s \in S$ denote the final set-fused element containing $e$, obtained after the termination of the While-loop. The element $s$ receives exactly one color, which is then assigned to all elements obtained by set-unfusing $s$, including $e$.

    We next show that $c$ is rainbow-circuit-free. Observe that, before any set $L \in \mathcal{L}$ is removed from $\mathcal{L}$, the number of elements $s \in S$ that contain elements of $L$ is reduced to at most $b(L)$. Consequently, at most $b(L)$ distinct colors are assigned to elements of $L$. 
    It then follows from \Cref{lem:rainbowlaminar} that the coloring $c$ is rainbow-circuit-free.
    
    Finally, since the sum of growth rates in any set $L \in \mathcal{L}$ remains invariant throughout the execution of the algorithm until $L$ is removed from $\mathcal{L}$, to prove that each color class has total growth rate less than~$2$, it suffices to show that whenever two elements $s_1,s_2$ are set-fused into $s_1 \hatcirc s_2$, we have $g(s_1 \hatcirc s_2) = g(s_1) + g(s_2) < 2$.
    
    To this end, consider two elements $s_1$ and $s_2$ that are set-fused into $s_1 \hatcirc s_2$ at Line~8, and assume without loss of generality that $g(s_1) \leq g(s_2)$. Recall that initially $\sum_{e \in L} g(e) \leq b(L)$ holds for every $L \in \mathcal{L}$, and it is easy to see that this inequality is preserved throughout the execution of the algorithm. Moreover, since $|L| > b(L)$ at the time of the fusion, we have
    \[
    \sum_{e \in L} g(e) \leq b(L) \leq |L| - 1.
    \]
    As $s_1$ and $s_2$ are chosen to minimize $g(s_1)+g(s_2)$, it follows that
    \[
    g(s_1) \leq \frac{b(L)}{|L|} < 1
    \quad\text{and}\quad
    g(s_2) \leq \frac{b(L) - g(s_1)}{|L| - 1} \leq \frac{b(L)}{|L| - 1} \leq 1.
    \]
    Therefore, $g(s_1) + g(s_2) < 2$, as claimed.

    For the running time, for each $L \in \mathcal{L}$, we set-fuse elements in $L$ at most $|L|-b(L)\leq |E|-b(L)$ times, and $\mathcal{L}$ contains $\operatorname{poly}(|E|)$ sets; the implementation of set-fusion and set-unfusion can be implemented in polynomial time as in \Cref{alg:fun}. The rest of the operations can be done in polynomial time. Hence, $c$ can be computed in polynomial time.
\end{proof}

\thmLaminarFUN*
\begin{proof}
    The schedule $\pi$ is constructed from $k$ disjoint Fuse--Unfuse schedules $\pi_i$ for some $k$.
    By \Cref{thm:rainbowfreelaminar}, we obtain a rainbow-circuit-free coloring $c$ of $M$ such that each color class $E_i$ satisfies that the sum of growth rates is less than $2$, $i\in[k]$ for some $k$. For each $i\in[n-1]$, let $\pi_i$ be the Fuse--Unfuse schedule for the instance $(E_i,g_i)$ where $g_i$ is the growth rate function induced by $g$ on $E_i$. We define $\pi$ as follows: At time $t$, cut the set $I_t=\cup_{i \in[n-1]}\pi_i(t)$. As we pick an element of each color class, and the coloring is rainbow-circuit-free, we obtain that $I_t$ is independent, and thus the schedule is valid.
    
    For the height, observe that by \Cref{cor:bgt4} we have that $h(\pi_i) < 4$, and hence $h(\pi) = \max_i h(\pi_i) < 4$ as required.

    Finally, the coloring $c$ can be computed in polynomial time by \Cref{thm:rainbowfreelaminar}, and the schedule for each of the $k$ instances can be constructed in polynomial time using \Cref{alg:fun}, where $k\leq|E|$. 
    Consequently, the resulting schedule can be implemented in polynomial time, which completes the proof.
\end{proof}

\section{Improved Lower Bound for General Set Systems}
\label{app:jamison-lb}

Given $v \in \{0,1\}^k$, the affine hyperplane (avoiding $\mathbf{0}$) defined by $v$ is 
\[
H_v = \{u \in \{0,1\}^k \mid v \cdot u \equiv 1 \pmod 2 \}.
\]
We define the set $\cH$ as the set of all affine hyperplanes, i.e., 
\[
\cH_k = \{H_v \mid v \in \{0,1\}^k \setminus \{\mathbf{0}\}\}.
\]
The following theorem of Jamison \cite{MR439664} gives a lower bound on the number of affine hyperplanes covering the cube (except $\mathbf{0}$).
 
\begin{thm}[Theorem $1'$ in \cite{MR439664}]
    \label{thm:jamison}
    Let $H_1,\dots,H_m \in \cH_k$ be a set of affine hyperplanes covering $\{0,1\}^k \setminus \{\mathbf 0\}$.
    Then, $m\geq k$.
\end{thm}

Using \Cref{thm:jamison} as a source of a hard-to-cover set system, we can prove \Cref{thm:jamison-lower-bound}

\jamisonlb*

\begin{proof}
    Let $E = \{0,1\}^k \setminus \{\mathbf 0\}$, and we consider the independent sets $\cI =\cH_k$.
    We consider now the growth vector given by the coefficients $\lambda(H_v) = \frac{1}{|\cH_k|}$.
    This is a valid growth vector as $\sum_{v \in \{0,1\}^k \setminus \{\mathbf{0}\}} \frac{1}{|\cH|}=1$.
    Furthermore, for every $x \in \{0,1\}^k \setminus \{\mathbf{0}\} $ let
    \[
    g(x) = \sum_{v \in \{0,1\}^k \setminus \{\mathbf{0}\} \colon x \in H_v} \frac{1}{|\cH_k|} = \frac{|\{v \in E \mid  v \cdot x \equiv 1 \pmod 2 \}|}{|\cH_k|}= \frac{|H_x|}{|\cH_k|} = \frac{2^{k-1}}{2^k-1} > \frac{1}{2}.
    \]

    Now, for every schedule, there exists an element $v^*\in E$ that is not cut during the first $k-1$ time steps.
    Otherwise, we have that the sets $H_{v_1},\dots,H_{v_{k-1}} \in \cH_k$ selected by the schedule would cover $\{0,1\}^k\setminus \{\mathbf{0}\}$, contradicting \Cref{thm:jamison}.

    Recalling that $k= \log_2 (n+1)$, we have that the height of $v^*$ at time $k-1$ is 
    \[
    h_{k-1}(v^*) = (k-1)g(v^*) > \frac{k-1}{2} = \frac{\log_2 (n+1)-1}{2}.\qedhere
    \]
\end{proof}

\end{document}